\documentclass{IEEEtran}
\IEEEoverridecommandlockouts
\usepackage{amsmath,amssymb,amsfonts}
\usepackage[cmintegrals]{newtxmath}
\usepackage{caption}
\ifCLASSOPTIONcompsoc,
\usepackage[caption=false, font=normalsize, labelfont=sf, textfont=sf]{subfig}
\else
\usepackage[caption=false, font=footnotesize]{subfig}
\usepackage{tabularx}

\newtheorem{proposition}{Proposition}

\newtheorem{proof}{Proof}

\usepackage{mathtools}
\usepackage{algorithm}
\usepackage{algorithmic}
\usepackage{graphicx}
\usepackage{textcomp}
\usepackage{xcolor,color}
\usepackage{multirow}
\usepackage{subfig}
\usepackage{siunitx}
\usepackage{booktabs}
\usepackage{enumerate}
\usepackage{fixltx2e}
\usepackage{float}
\usepackage{graphicx} 
\usepackage{marvosym}
\usepackage[noadjust]{cite}
\usepackage{soul}
\usepackage{stfloats}
\usepackage{cite}

\usepackage[numbers,sort&compress]{natbib}
\makeatletter
\newcommand{\removelatexerror}{\let\@latex@error\@gobble}
\makeatother

\def\BibTeX{{\rm B\kern-.05em{\sc i\kern-.025em b}\kern-.08em
    T\kern-.1667em\lower.7ex\hbox{E}\kern-.125emX}}
    
\ifodd 1
\usepackage{soul}
\usepackage{color}
\setstcolor{red}

\newcommand{\del}[1]{\st{#1}} 

\newcommand{\com}[1]{\textbf{\color{red} (COMMENT: #1)}} 
\newcommand{\response}[1]{\textbf{\color{blue} (RESPONSE: #1)}} 
\else

\newcommand{\del}[1]{}

\newcommand{\com}[1]{}
\newcommand{\comg}[1]{}
\newcommand{\response}[1]{}
\fi

\begin{document}

\title{Impact and Analysis of Space-Time Coupling on Slotted MAC in UANs}

\author{Yan Wang,
Quansheng Guan,~\IEEEmembership{Senior Member,~IEEE, }
Fei Ji,~\IEEEmembership{Member,~IEEE, } 
Weiqi Chen
\IEEEcompsocitemizethanks{\IEEEcompsocthanksitem Y. Wang, Q. Guan, and F. Ji are with the School of Electronic and Information Engineering, South China University of Technology, Guangzhou 510640, China\protect\\
\IEEEcompsocthanksitem W. Chen is with the School of Internet Finance and Information Engineering, Guangdong University of Finance, Guangzhou 510000, China.}
}

\maketitle

\begin{abstract}


The propagation delay is non-negligible in underwater acoustic networks (UANs) since the propagation speed is five orders of magnitude smaller than the speed of light. In this case, space and time factors are strongly coupled to determine the collisions of packet transmissions. To this end, this paper analyzes the impact of space-time coupling on slotted medium access control (MAC). 
We find that a sending node has specific location-dependent interference slots and slot-dependent interference regions. Thus, the collisions may span multiple slots, leading to both inter-slot and intra-slot collisions. Interestingly, the slot-dependent interference regions could be an annulus inside the whole transmission range. It is pointed out that there exist collision-free regions when a guard interval larger than a packet duration is used in the slot setting. In this sense, the long slot brings spatial reuse within the transmission range. However, we further find that the guard interval is not larger than a packet duration, which is much shorter than the existing slot setting in typical Slotted-ALOHA, to reach a peak successful transmission probabilities and throughput. 
Simulation results verify our findings, and also show that the performance of vertical transmissions is more sensitive to the spatial impact than horizontal transmissions in UANs.
\end{abstract}


\section{Introduction}
{U}{nderwater} acoustic networks (UANs) are envisioned widely for oceanic monitoring, fisheries activities, ecological protection, and other commercial or scientific applications, which require multiple underwater sensors to sense the underwater environment and report the data to the sink on the water surface \cite{Ref-UAN-Phy,Ref-UANs,Ref-SmartOcean}. When the data packets from sensor nodes arrive at a receiver simultaneously, these packets will collide with each other and their reception fails. In this case, medium access control (MAC) is critical to enable multiple sensor nodes to share the open channel in UANs. 

The low speed of acoustic waves and the long propagation delay make the MAC in UANs different from that in terrestrial radio networks (TRNs) \cite{Mandar-Throughput,Ref-Jiang,Ref-2021-Survey}. The speed of acoustic waves (i.e., 1500 m/s) is five orders of magnitude slower than the speed of radio waves (i.e., $3\times 10^8$ m/s). In TRNs, the propagation delay of radio waves can be neglected, and the arriving time of a packet is approximated to its sending time. The principle for MAC is to avoid concurrent sending in order to guarantee collision-free transmissions. However, the low speed of underwater acoustic waves introduces a non-negligible propagation delay to packet transmissions in UANs \cite{Ref-Jiang}. The packet arriving time at the receiver is then determined by both spatial distances between nodes and sending times of packets \cite{Ref-ST-Coupling}. We can observe on the one hand that packets who are simultaneously sent from different nodes might not arrive at the same time at the receiver. On the other hand, the packets might arrive at the same time although their sending times are different. To this end, the MAC strategies in TRNs like CSMA, CDMA, etc. are not sufficient to avoid the space-time coupling collisions in UANs \cite{Ref-BiC-MAC,Ref-ST-MAC,Mandar-Variable-Duration,Ref-LiuMeiyan,Ref-QlearningMAC,Ref-UW-SEEDEX}.

Such space-time coupling in UANs introduces not only time uncertainty but also space uncertainty to random access based MAC \cite{Ref-Affan}.
Since the generation of the data is often considered as a random process in communication networks \cite{Data-Network,Ref-Ahn}, the sending time is uncertain to the destination node. 
Similarly, the positions of nodes are not available in UANs due to the lack of positioning systems, and the sensor nodes are considered as randomly deployed \cite{Ref-Radom-Network,Ref-Zhong}. The spatial propagation delays are also uncertain to the destination node. Due to the ignorable propagation delay, the MAC protocols in TRNs consider only the time uncertainty. 

The slotting technique could be used to alleviate the time uncertainty. It divides the channel time into slots, and nodes are only allowed to access the channel at the beginning of a slot. By ignoring the propagation delay, the slot length is often set to the duration time of one packet transmission in TRNs. It has been proved that the slotting technique could almost double the network throughput of TRNs \cite{Data-Network}. 

Unfortunately, it was reported in literature that the slotting technique does not improve the performance of MAC in UANs \cite{Ref-Vieira}. 
Due to the long propagation delay, the transmission in the current slot may arrive at the following several slots at the receiver, which brings \emph{inter-slot collisions} to UANs. 
Furthermore, the propagation delays between nodes are uncertain, i.e., space uncertainty. An extra guard interval is demanded in a slot to accommodate this kind of space uncertainty. 
To ensure inter-slot collision-free transmissions for all the nodes in the network, the slot length has to be one packet duration plus the maximal propagation delay between nodes as a guard interval  \cite{Ref-Ahn,Ref-LiuMeiyan,Ref-QlearningMAC}. 
Considering the long communication range and low propagation speed in UANs, the maximal propagation delay is always multiple times of a packet duration \cite{Ref-UW-SEEDEX} which is much longer than the slot length in TRNs (i.e., the duration of a packet transmission). 


Although a long slot can decouple space and time, it degrades the MAC performance. 
On the one hand, a longer slot will accumulate more data packets that have to compete to access the next slot, aggravating the channel competition. The \emph{intra-slot collision} probability will increase in this case. 
On the other hand, a long slot means a long idle waiting, which further decreases the sending opportunities and the throughput of MAC. 

It is noticed that the concurrent transmissions may not lead to collisions at the receiver due to the space-time coupling feature of UANs \cite{Ref-BiC-MAC,Ref-ST-MAC,Mandar-Variable-Duration}. 
In this sense, the existing slot length setting (i.e., one packet duration plus the maximal propagation delay) is unnecessary protection from transmission collisions. 
It is possible to shorten the slot length to increase the MAC performance intuitively. However, the slot length setting and collision relate strongly to the space-time coupling factors, such as the guard interval, transmission duration, setting time, node positions, etc. Therefore, this paper wants to understand the impact of space-time coupling and understand the optimal slot setting for slotted MAC. 

The contributions and findings of this paper are summarized as follows. 
\begin{itemize}
\item \emph{Analysis of space-time coupling collisions:} We derive the collision span of slots and the slot-dependent interference regions in UANs. It is found that the interference region could be an annuls for a slot. Therefore, collision-free regions may exist in the transmission range of a receiver, and space reuse is possible to accommodate concurrent transmissions. 
\item \emph{Closed-form expressions for successful transmission probability and throughput:}  We establish the relationship between successful transmission probability and the coupling factors. We find that the successful transmission probabilities and throughput are the same for the slot lengths of one packet duration and of two packet durations. Therefore, the maximizer of the slot length exists in between as shown in the simulations. We also find that the spatial factor has a greater impact on vertical transmissions than horizontal transmissions in UANs.
\end{itemize}

The remainder of this paper is arranged as follows. Section \ref{sect:related} overviews the related work, and Section \ref{sect:system} presents the system model and the motivation of this work. The space-time coupling collisions are discussed in Section \ref{sect:collisions}. The successful transmission probability and throughput are then analyzed in Section \ref{sect:perf}. Simulation results using NS3 are used to verify our analysis in Section \ref{sect:sim}. Finally, Section \ref{sect:con} concludes this paper. 

\section{Related work}\label{sect:related}
The idle waiting is the major challenge brought by the long propagation delay in the MAC design for UANs.
Many efforts tried to exploit the idle waiting to increase the channel utilization.
Handshaking is often used to reach a consensus among a group of nodes to reuse idle waiting.
Bidirectional Concurrent MAC (BiC-MAC) enabled bidirectional concurrent transmissions which allow the receiver to transmit packets during idle waiting \cite{Ref-BiC-MAC}. 
In the Reverse Opportunistic Packet Appending (ROPA) protocol, the handshake's initiator invites its neighbors to opportunistically transmit their packets after the initiator's transmission \cite{Ref-ROPA}.

Transmission scheduling is another efficient method to minimize the idle waiting.
Hus \emph{et al.} used the delay information to construct a Spatial-Temporal Conflict Graph (ST-GC) to describe the conflict delays among links.
The MAC problem was transformed into a vertex coloring problem of ST-GC and the optimal scheduling could be obtained \cite{Ref-ST-MAC}.
Anjangi \emph{et al.} formulated the scheduling into a Mixed Integer Linear Fractional Programming (MILFP) problem, which adjusts the sending time and packet duration to minimize the frame length \cite{Mandar-Variable-Duration, Mandar-Throughput}.
Since a packet may propagate for several scheduling cycles, Chen \emph{et al.} further considered the collisions among scheduling cycles in the formulated mixed-integer linear programming
problem \cite{Ref-ST-Coupling}.

Analysis was also conducted to understand the impact of long propagation delay on underwater MAC.
The simulation results showed that the throughput of Slotted-ALOHA is the same as Pure-ALOHA due to the long propagation delay in UANs \cite{Ref-Vieira}, while the throughput of Slotted-ALOHA is double compared to that of Pure-ALOHA in TRNs \cite{Data-Network}. 
Ahn \emph{et al.} pointed out that it is the space uncertainty introduced by the varying propagation delays that degrade Slotted-ALOHA in UANs, and analyzed the throughput in a case that the packet transmission duration is larger than the propagation delay in \cite{Ref-Ahn}. 
However, it is noticed that the propagation delay over long-range underwater acoustic links is potentially larger than the transmission duration in practical UANs, and the propagation may span several slots.  

Other MAC protocols, like floor-acquisition multiple access (FAMA) and distance-aware collision avoidance protocol (DACAP), tried to apply the strategy of multiple access with collision avoidance (MACA) in TRNs into UANs.  Shahabudeen \emph{et al.} used a Markov-chain model to analyze the impact of batch queuing on the MACA and gave closed-form expressions for the mean service time and throughput \cite{Ref-MACA-analysis}. 
Zhu \emph{et al.} studied the impact of the low transmission rates and long preambles of the real acoustic modem on the random access-based MAC and handshake-based MAC.
In order to capture the effects of space-time uncertainty, Guan \emph{et al.} purposed a statistical physical interference model which is based on the observed interferences over passed slots \cite{Ref-Statistic-Interference}. 
This statistical model was validated by testbed experiments and used by each node to distributively select an optimal access probability, which increases the sum-throughput over multiple slots.
Chen \emph{et al.} further revealed that the collisions of  space-time coupling transmissions in UANs are different from that in TRNs \cite{Ref-ST-Coupling}. 
It was shown that there exist possibly multiple concurrent successful transmissions due to the coupling in space and time. 

This paper further discusses the spatiotemporal impact on slotted-MAC in UANs. 
We try to answer the questions that whether space-time coupling brings performance gain to slotted MAC and what is the optimal slot length to maximize the successful transmission probability and throughput.

\section{System Model and Motivation}\label{sect:system}
This section describes the system model and the motivation of our work.
\subsection{System Model}
We consider a network that has one sink (e.g., a buoy on the water surface) to collect data from $N$ underwater sensor nodes. 
Let $\mathbf{N}=\{1, {2}, \ldots, {N}\}$ denote the set of sensor nodes. 

Underwater sensor nodes adopt a slotted MAC to share the open channel. 
Packets that arrive in current slot are stored in the buffer and sent at the beginning of the next slot. 
Each slot accommodates one packet duration and one additional guard interval. 
Let $t_{slot}$ denotes the duration of a slot. The expression of the slot length is written by 
\begin{equation}
t_{slot}=t_{f}+\beta\!\cdot\!\tau, 
\label{equa-general-slot-length}
\end{equation}  
where $t_f$ is the transmission duration for a packet, $\tau$ denotes the maximal propagation delay, and $\beta$ denotes the guard coefficient, which is the ratio of the guard interval to the maximal propagation delay.

\subsection{Motivation}\label{sect: motivation}
The length of a slot affects the network performance directly. Let us look at two cases at $\beta \geq1$ and $\beta<1$, respectively. 

When $\beta \geq1$, the slot length is not shorter than the transmission time (i.e., $t_f$) plus the maximum propagation time (i.e., $\tau$). 
In this case, one packet can be transmitted from the sender to its receiver within one slot. 
According to the principle of slotted MAC that packets are only allowed to be sent at the beginning of a slot, collisions are impossible for any packets that are sent in different slots. 
Collisions may happen only for packets sent in the same slots, which is called \emph{intra-slot collisions}. 

It would be much more complicated for the case of $\beta<1$. 
A slot length having $\beta<1$ cannot accommodate the complete transmission of a packet, and a packet transmission may propagate across several slots. 
In addition to intra-slot collisions, packets from different slots may also collide with each other, i.e., \emph{inter-slot collisions}, as shown in Fig. \ref{fig-motivation-examples2}, which is the main reason that degrades Slotted-ALOHA to Pure-ALOHA \cite{Ref-Vieira}. 

\begin{figure}
  \centering
  \subfloat[Slotted-MAC with $\beta\geq 1$]{
    \label{fig-motivation-examples1}
  \includegraphics[width=0.8\columnwidth]{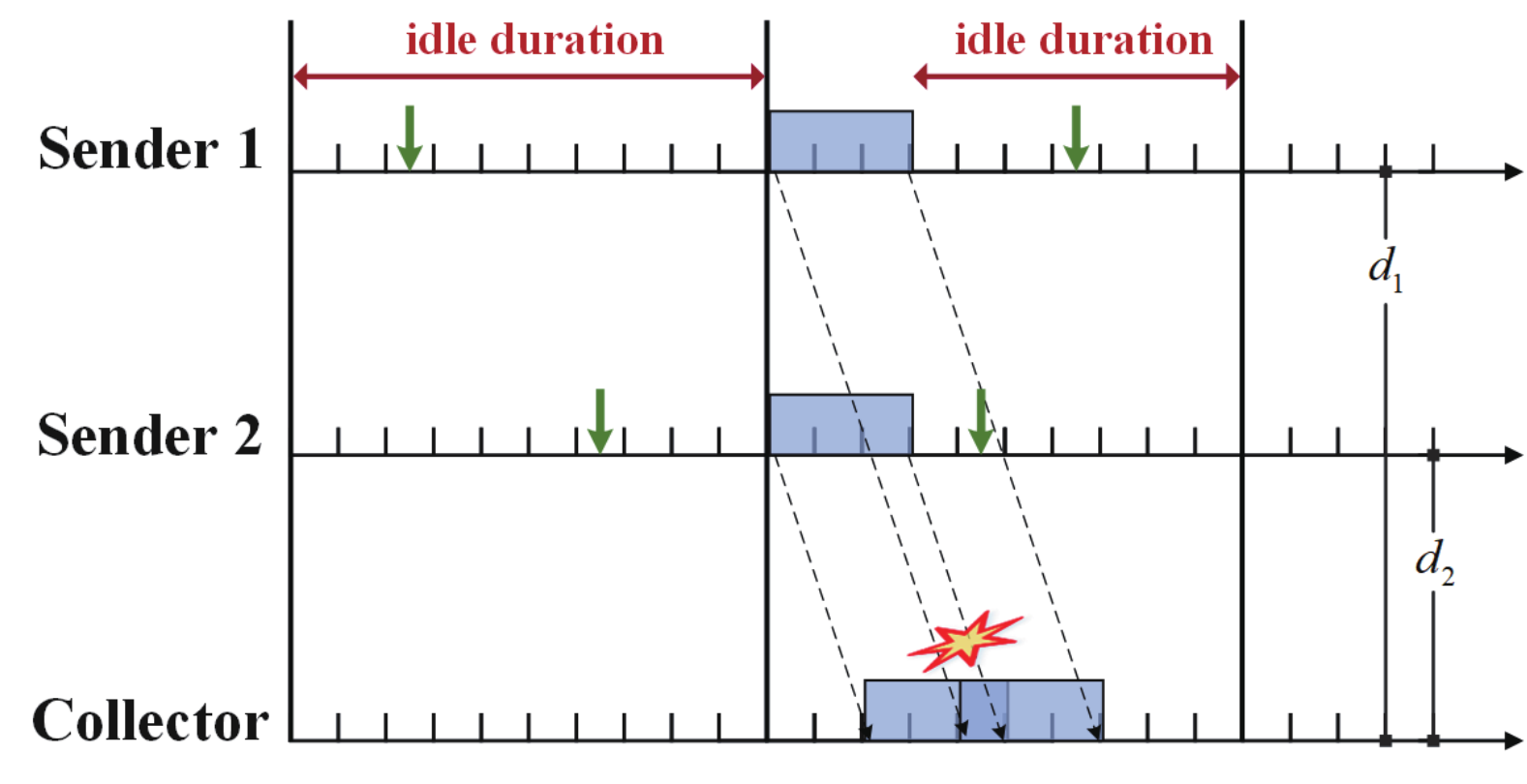}}
\\
  \subfloat[Slotted-MAC with $\beta<$1]{
    \label{fig-motivation-examples2}
  \includegraphics[width=0.8\columnwidth]{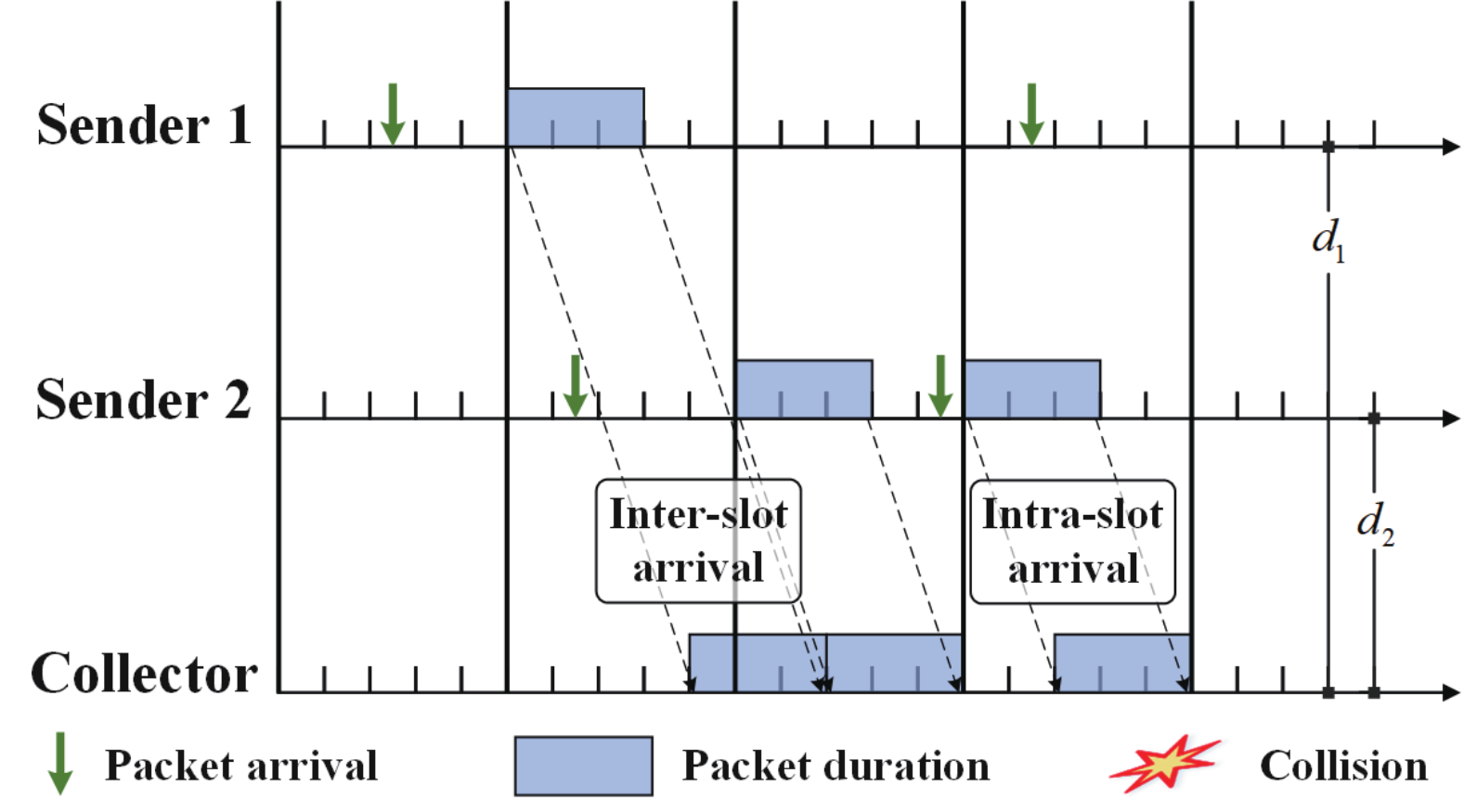}}
  \caption{Transmission collisions in slotted-MAC. (a) When $\beta\geq 1$, only intra-slot collisions exist for packets sent from two senders at distances of $d_1$ and $d_2$, respectively). However, more packets will be accumulated in one slot, compared to the case of $\beta<$1. (b) Both intra-slot and inter-slot arrivals may not exist without collision when $\beta<$1. 
}
  \label{fig-motivation-examples}
\end{figure}

A relatively long slot, e.g., $\beta\geq1$, eliminates inter-slot collisions by the long guard interval. However, Fig. \ref{fig-motivation-examples1} shows that the long slot length may still decrease MAC performance in two aspects. Since transmissions are only allowed at the beginning of a slot, a long guard interval means a long idle waiting between packet transmissions, which decreases the transmission rate. 
In addition, although the long slot eliminates inter-slot collisions, it aggravates intra-slot collisions. 
More packets may arrive during a longer slot and will be accumulated to compete for channel access at the next slot. 
Thus, the intra-slot collision probability will increase for these backlogged packets. 
The low transmission rate and the high collision probability will finally lead to low network throughput. 

Although a slot having $\beta\geq 1$ could eliminate space uncertainty and inter-slot collisions, the space-time coupling slot setting having $\beta<1$ may achieve better performance. The space-time coupling feature in UANs motivates our study in this paper. The notations in Tab. \ref{table-notations} are used in our discussions. 
\begin{table}[!ht]
  \caption{The notations in the discussion.}
  \centering
  \setlength{\tabcolsep}{0.55mm}{
  \begin{tabular}{c|c} 
  	\hline 
  	\textbf{Symbols} & \textbf{Meaning}\\
  	\hline
  	$N$ & The number of sensors\\
  	\hline
  	$\mathbf{N}$ & The set of sensors\\
  	\hline
  	$t_{f}$ & The packet duration\\
  	\hline
  	$t_{slot}$ & The slot duration\\
  	\hline
  	$\tau$ & The maximal propagation delay  in the horizontal plane\\
  	\hline
  	$\beta$ & The guard coefficient\\
  	\hline
  	$t_{i}$ & The sending time of sensor $i$\\
  	\hline
  	$d_{i}$ & The distance from sensor $i$ to the sink\\
  	\hline
  	$v$ & The propagation speed of underwater acoustic wave\\
  	\hline
  	$\Delta m$ & The sending slot of an interference node\\
  	\hline
  	$M$ & The maximal number of interference slots\\
  	\hline
  	$\mathbf{M}$ & The set of slot difference to possible interfering slots\\
  	\hline
  	$R$ & The maximal communication range in the horizontal plane\\
  	\hline
  	$K$ & The amount of distance segments\\
	\hline
	$\mathbf{D_{k}}$ & The $k$-th distance segment \\
  	\hline
  	$\mathbf{D}$ & The set of distance segments\\
  	\hline
  	$\mathbf{C_{k}}$ & The set of interference slots for a tagged node in $\mathbf{D_{k}}$\\
  	\hline
  	$\lambda$ & The packet arrival rate per node\\
  	\hline
  	$\alpha$ & The expanding coefficient for the vertical communication range\\
    \hline
    $\text{IR}_i^{\Delta m}$ & The interference region of slot $\Delta m$ for a node $i$ in $\mathbf{D_{k}}$\\
    \hline
    $S_{1,k}^{\Delta m}$ & The area of $\text{IR}_i^{\Delta m}$ for node $i$\\
    \hline
    $S_{2,k}^{\Delta m}$ & The area of DIR of Slots $\Delta m$ and $\Delta m+1$ for node $i$\\
    \hline
    $P_{z,k}$ & The probability that an interference node falls in IRs of node $i$\\ 
    \hline
    $P_{o,k}$ & The probability that an interference node falls in DIRs of node $i$\\ 
    \hline
    $P_{t}$ & The sending probability of a node\\
    \hline
    $P_{k}$ & The probability that a node locates in $\mathbf{D_{k}}$\\
    \hline
  \end{tabular}}
  \label{table-notations}
\end{table}

\section{Space-Time Coupling Collisions in Slotted MAC}\label{sect:collisions}
We adopt the protocol transmission model to describe whether packet transmissions collide  \cite{Ref-ST-Coupling, Ref-Zhong, Gupta-WN-capacity}.
Suppose node ${i}$ sends a packet at time $t_{i}$, and $d_{i}$ denotes the distance between node ${i}$ and the sink.
The transmission of node ${i}$ does not collide with the transmission of node ${j}$ (${j} \in \mathbf{N}/\{i\}$) at the sink only when it satisfies  
\begin{equation}
\left|t_{i}-t_{j}+\frac{d_{i}-d_{j}}{v}\right| \geq t_f, 
\label{equa-protocol-model}
\end{equation}  
where $d_{j}$ is the distance between node ${j}$ and the sink, and node ${j}$ sends a packet at time $t_{j}$. According to \eqref{equa-protocol-model}, it depends on the sending slots (i.e., $t_i$ and $t_j$), sender locations (i.e., $d_i$ and $d_j$), and packet duration (i.e., $t_f$) to determine whether two transmissions will collide. 

According to the slotting technique, the difference between $t_{i}$ and $t_{j}$ should satisfy
\begin{equation}
  t_{i}-t_{j}= \Delta m{\cdot}t_{slot}, \Delta m \in \mathbb{Z}.
  \label{equa-time-differences}
\end{equation}
When node $i$ and node ${j}$ transmit packets at the same slot, $\Delta m$ is equal to 0; otherwise, $\Delta m$ is a non-zero integer. When $\Delta m<0$, node $j$ sends at a slot after node $i$; while $\Delta m>0$, node $j$ sends at a slot before node $i$.


Due to the long propagation delay, the collision period may cross over several time slots. 
\begin{proposition}
The collision period of packet transmissions spans $2M+1$ slots at maximum in slotted MAC, where $M=\lceil \frac{\tau+t_{f}}{t_{slot}} \rceil-1$. For example, a transmission at the current slot (i.e., $\Delta m=0$) may collide with the transmissions in the slots $\Delta m\in \mathbf{M}$ and $\mathbf{M}= \{-M, \ldots,0, \ldots,M\}$.
\label{proposition1}
\end{proposition}
\begin{proof}
According to \eqref{equa-protocol-model} and \eqref{equa-time-differences}, two transmissions will collide with each other when 
\begin{equation}
  \left| \Delta m\cdot t_{slot}+\frac{d_{i}-d_{j}}{v}\right|< t_{f}.
  \label{equa-protocol-model-proof1}
\end{equation}
Then, we have  $\frac{-\frac{d_{i}-d_{j}}{v} - t_{f}}{t_{slot}} \!<\!\Delta  m$ or $\Delta m \!<\! \frac{\frac{d_{i}-d_{j}}{v} + t_{f}}{t_{slot}}$. Based on the definition of the maximum propagation delay, i.e., $|\frac{d_{i}-d_{j}}{v}|<\tau$, we have $\frac{-\tau - t_{f}}{t_{slot}} < \Delta m$ or $\Delta m < \frac{\tau + t_{f}}{t_{slot}}$. Let $M=\lceil \frac{\tau+t_{f}}{t_{slot}} \rceil-1$, we get $-M\leq \Delta m\leq M$, i.e., $\Delta m\in \mathbf{M}$. 
\IEEEQEDhere
\end{proof}

In Proposition \ref{proposition1}, $\Delta m=0$ corresponds to intra-slot collisions, while $\Delta m \neq 0$ corresponds to inter-slot collisions. 
Proposition \ref{proposition1} verifies our motivation in Section \ref{sect: motivation}. Based on Proposition \ref{proposition1}, when $\beta\geq 1$, $M=0$. The collisions can only happen in the current transmission slot, i.e, only intra-slot collisions exist. However, since $|\frac{d_{i}-d_{j}}{v}|<\tau$, many nodes have to idly wait for the end of the current slot to send packets. When $\beta< 1$, we have $M\neq 0$, and both inter-slot and intra-slot collisions exist. However, using a short slot length, the packet sending rate is increased and the waiting time is decreased. 


We next take an insight into the \emph{interference regions} (IRs) for space-time coupling transmissions.
\begin{proposition} 
For any sending node at a distance of $d_i$ away from the sink, its interference nodes may locate at the IRs in Fig. \ref{fig_horizontal_conflicit_region}.
\label{proposition2}
\end{proposition}

\begin{figure}
	\centering
	\subfloat[]{
	\includegraphics[width=0.22\columnwidth]{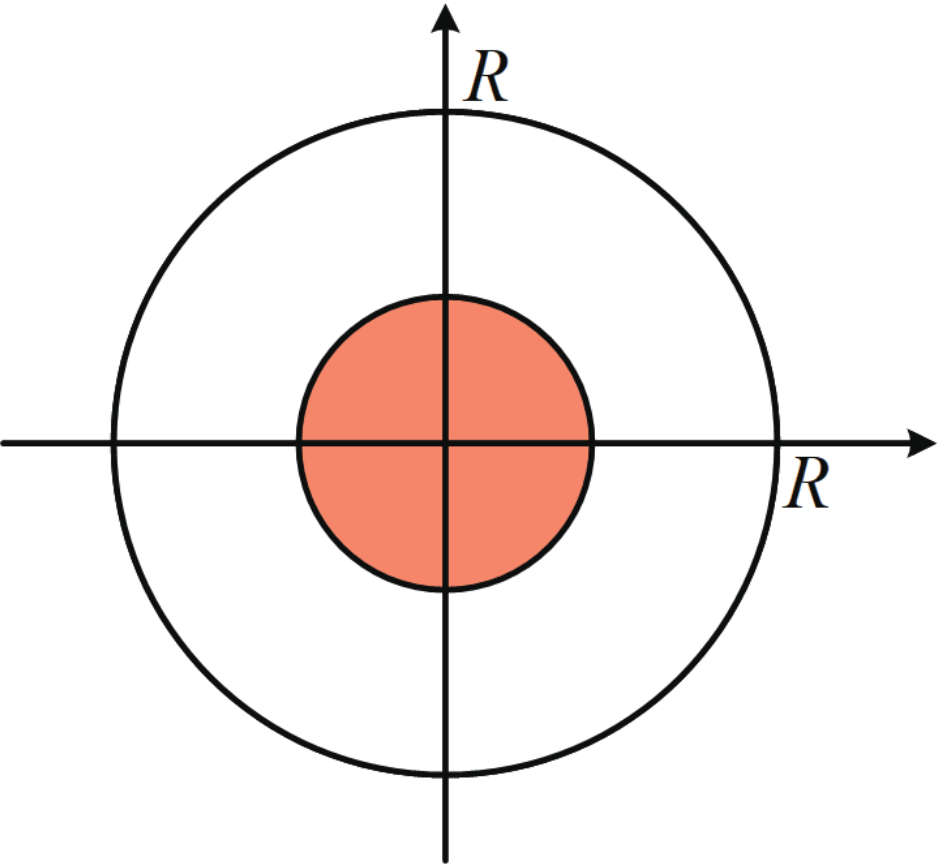}
	\label{fig-horizontal-a}}
	\hfill
	\subfloat[]{
	\includegraphics[width=0.22\columnwidth]{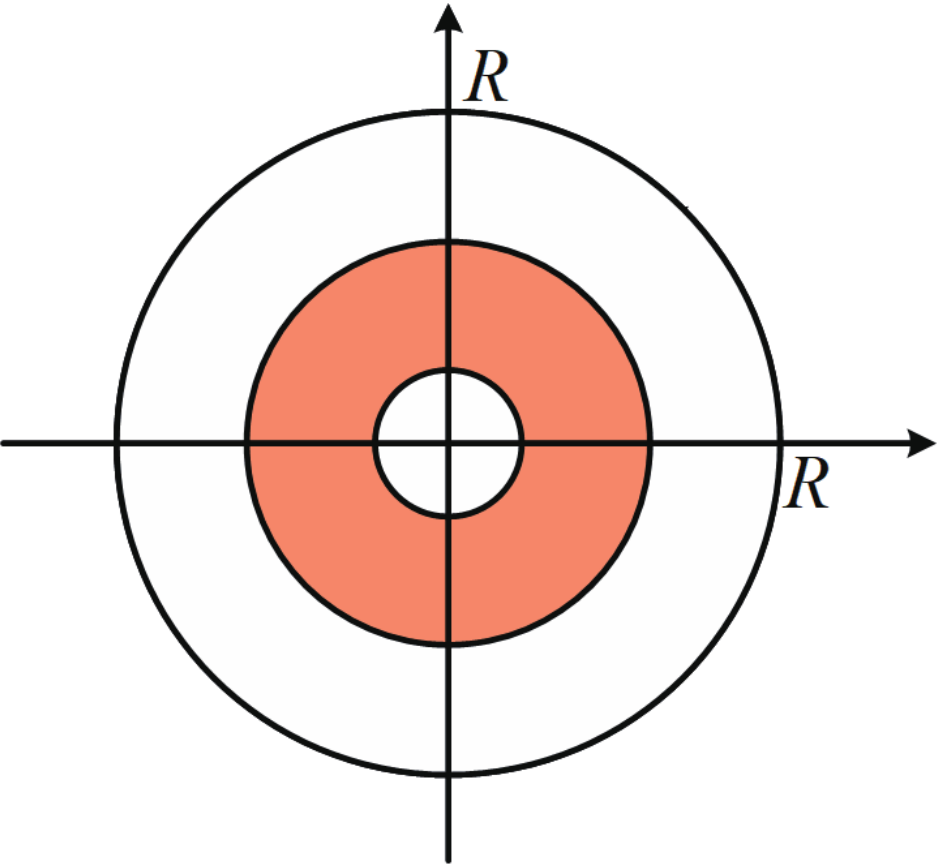}
	\label{fig-horizontal-b}}
	\hfill
	\subfloat[]{
	\includegraphics[width=0.22\columnwidth]{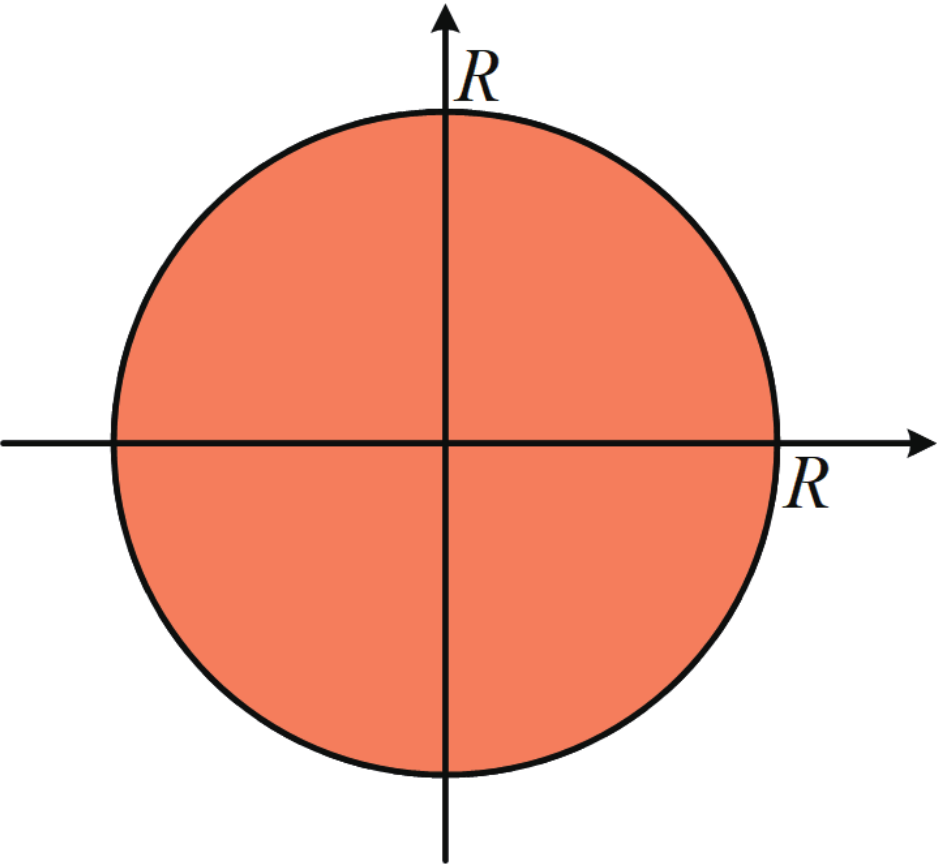}
	\label{fig-horizontal-c}}
	\hfill
	\subfloat[]{
	\includegraphics[width=0.22\columnwidth]{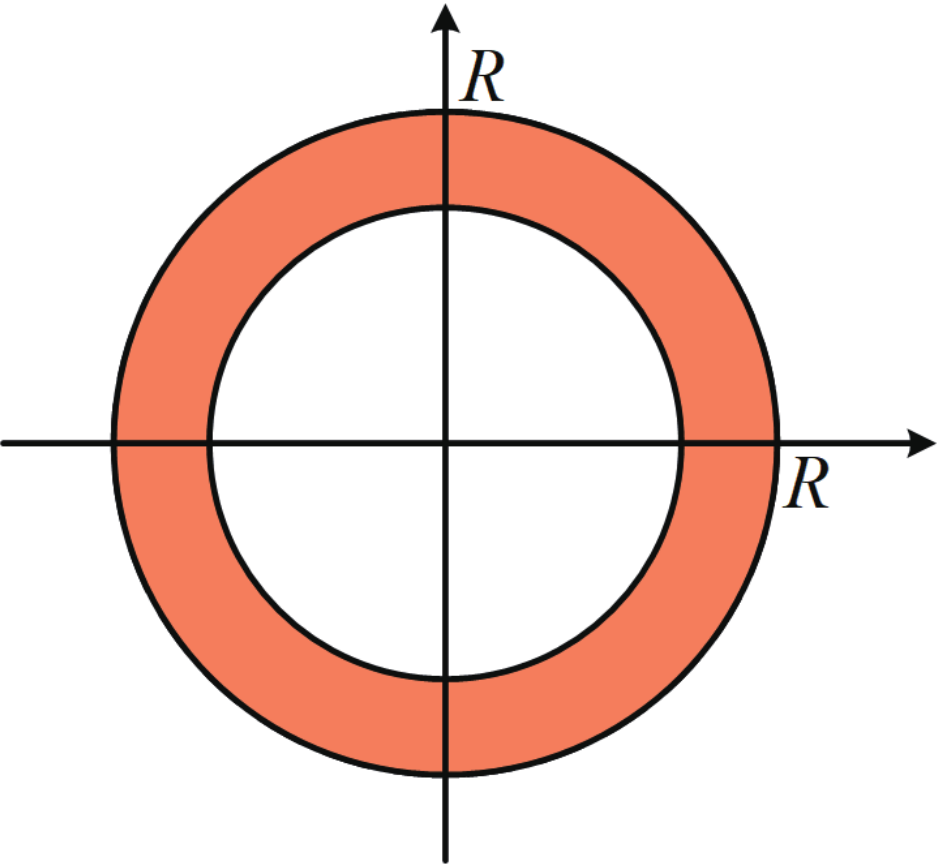}
	\label{fig-horizontal-d}}
	\caption{Four types of IRs for the space-time coupling collisions. 
	}
	\label{fig_horizontal_conflicit_region}
\end{figure}
\begin{proof}
Suppose the tagged node $i$ locates at $d_i$ and the coverage range of the sink is $R$. 
According to \eqref{equa-protocol-model-proof1}, the distances between the interference nodes and the sink for node $i$ should belong to  
\begin{equation}
  \text{IR}_i^{\Delta m} = [d_{i}+v(\Delta m\cdot t_{slot}-t_{f}), d_{i}+v(\Delta m\cdot t_{slot}+t_{f})]\cap [0, R].
  \label{equa-possible-interfering-region}
\end{equation}
The distance between an interference node and the sink should fall in $\text{IR}_i^{\Delta m}$, and the IRs belong to the following cases. 
\begin{itemize}
\item \emph{A circle}: When $d_{i}+v(\Delta m\cdot t_{slot}-t_{f})\leq 0$, the interference nodes situate in a circle around the sink with a radius of $d_{i}+v(\Delta m\cdot t_{slot}+t_{f})$ (c.f., Fig. \ref{fig-horizontal-a}). Specially, when $d_{i}+v(\Delta m\cdot t_{slot}+t_{f})>R$, all the nodes in the IRs can collide with node $i$ (c.f., Fig. \ref{fig-horizontal-c}). 
\item \emph{An annulus}: When $0< d_{i}+v(\Delta m\cdot t_{slot}-t_{f})\leq R$, only nodes in the annulus can collide with node $i$ (c.f., Fig. \ref{fig-horizontal-b} and Fig. \ref{fig-horizontal-d}), where the inside radius and the outside radius are $d_{i}+v(\Delta m\cdot t_{slot}-t_{f})$ and $\min (d_{i}+v(\Delta m\cdot t_{slot}+t_{f}), R)$, respectively. \IEEEQEDhere
\end{itemize}
\end{proof}

Proposition \ref{proposition2} shows that IRs depend on the sending slot (i.e., $\Delta m$). For example, when $\Delta m<0$, the colliding nodes locate around the sink within a circle. As $\Delta m$ increases, the IR becomes an annulus, which gets farther from the sink. Proposition \ref{proposition2} also shows that the radius of the circle or the width of the annulus is not larger than $2vt_{f}$. 

\begin{figure}
\centering
\includegraphics[width=\columnwidth]{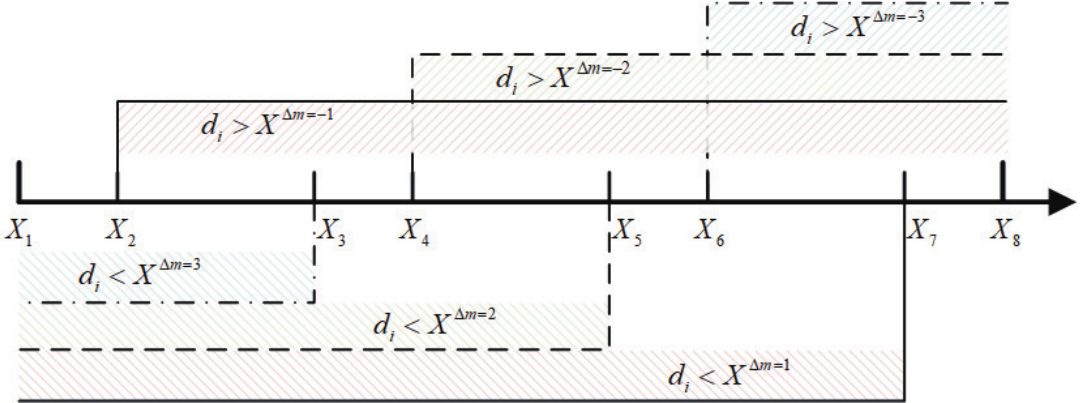}
\caption{Division of $\mathbf{D}$ according to \eqref{eq:Dk}. As an example, for $d_i \in \mathbf{D_{2}}=[X_2, X_{3}]$, its corresponding $\mathbf{C_{k}}=\{ -1,0,1,2,3 \}$.}
\label{fig_Dk_example}
\end{figure}

\begin{figure*}
    \centering
    \begin{minipage}[c]{1.5\columnwidth}
        \subfloat[]{
        \includegraphics[width=0.92\textwidth]{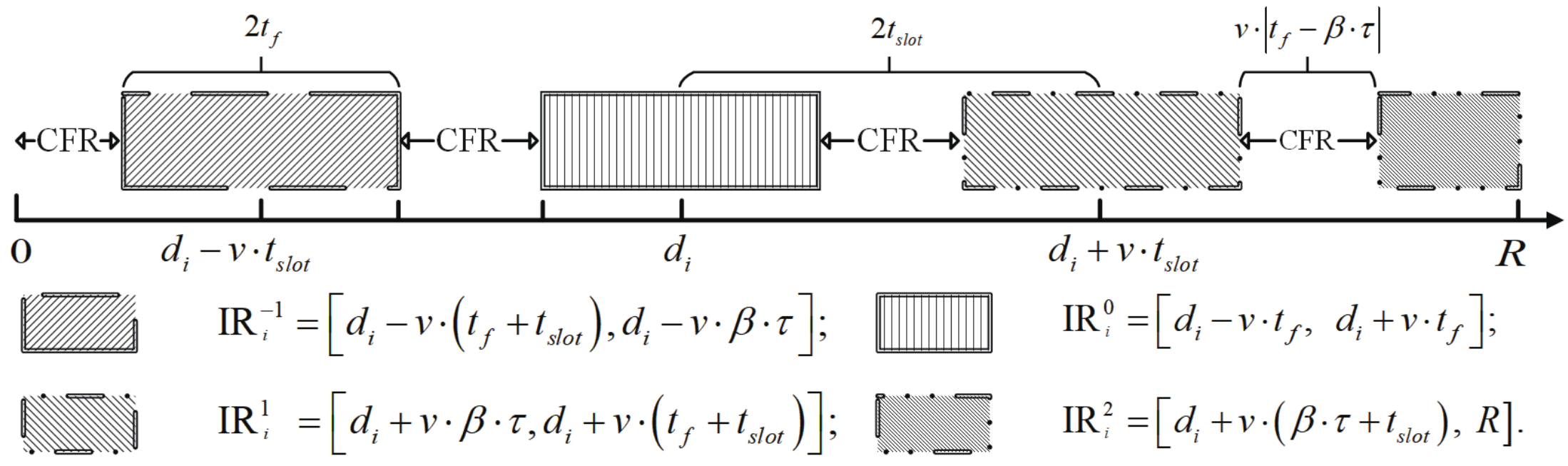}
        \label{fig-segement-demosration-small-pktduration}}
\\
        \subfloat[]{
        \includegraphics[width=0.92\textwidth]{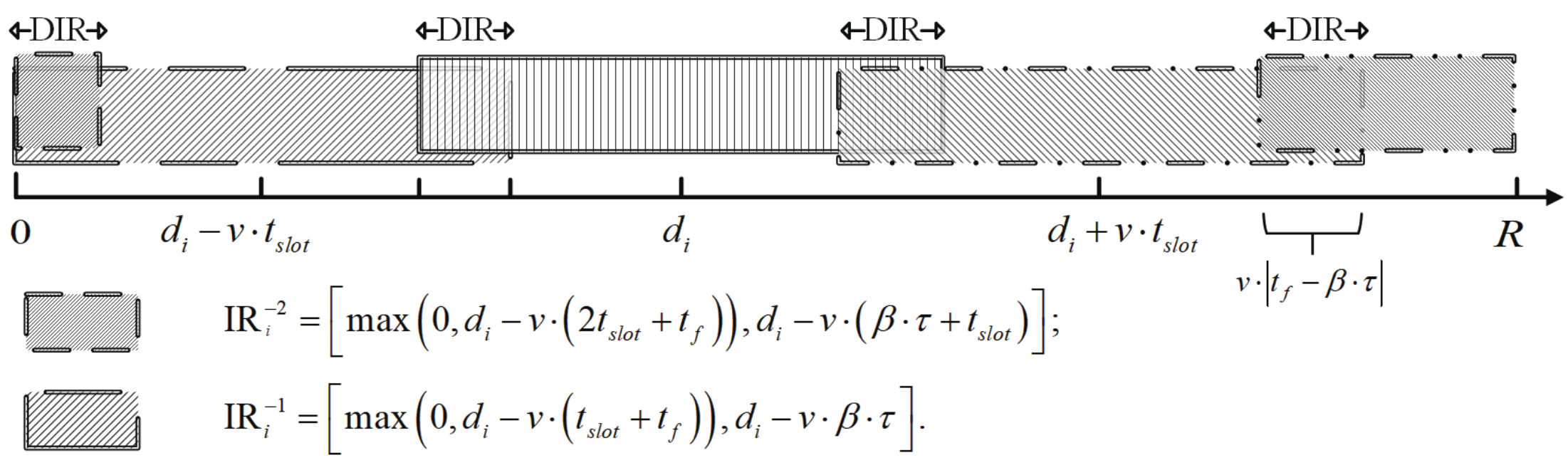}
        \label{fig-segement-demosration-big-pktduration}}
    \end{minipage}
    \begin{minipage}[c]{.5\columnwidth}
        \subfloat[]{
        \includegraphics[width=\textwidth]{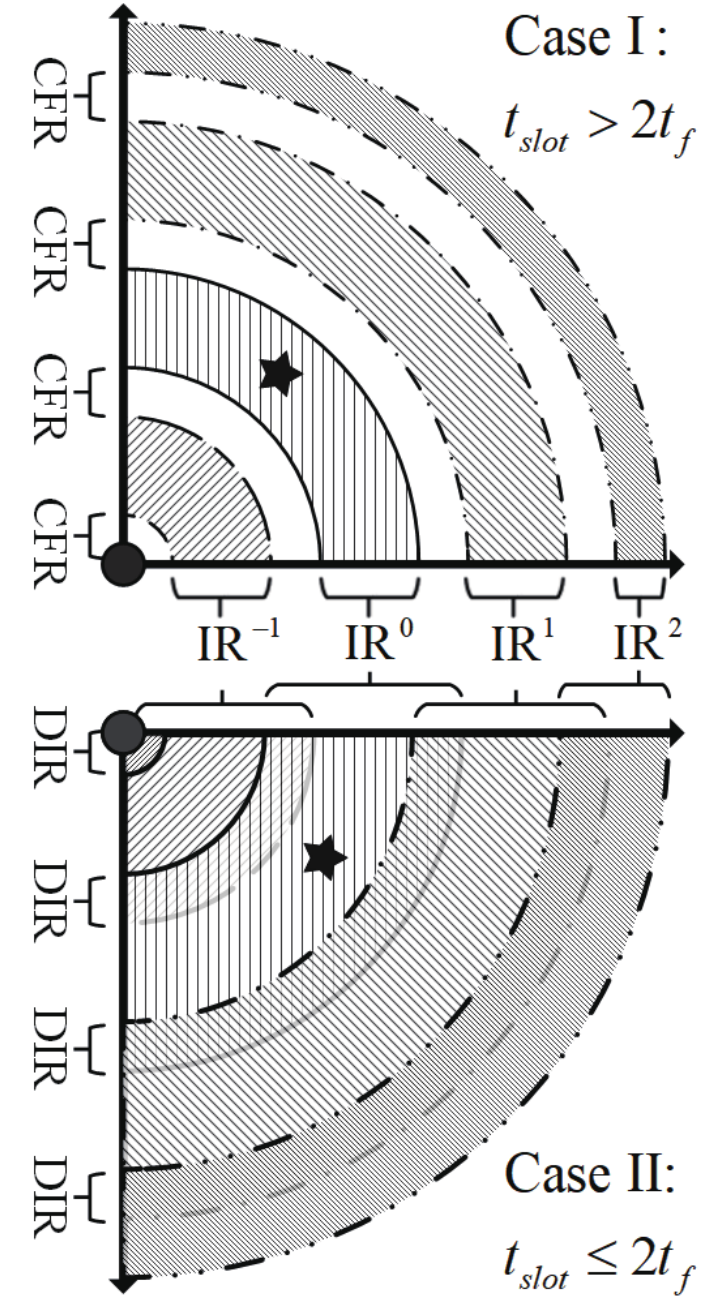}
        \label{fig-segement-demosration-cycle}}
    \end{minipage}
    \caption{IRs for different cases: (a) when $t_{slot}>2t_{f}  $, CFR between IRs exists; (b)  When $ t_{slot}\leq 2t_{f} $, DIRs exist due to the overlap of adjacent IRs. (c) The star stands for the tagged sending node. Each slot has a corresponding gray IR for this node. The blank gap between IRs is the CFR. The overlapped IR is the DIR. }
    \label{fig-interfering-segement}
\end{figure*}

Proposition \ref{proposition1} and Proposition \ref{proposition2} have shown that the collisions in UANs depend on the factors in both space and time dimensions. That is, a sending node has specific location-dependent interference slots and slot-dependent IRs. Therefore, we divide the communication range $[0,R]$ centered at the sink into $K$ segments, and nodes which locate at each segment have the same interference slots. Specifically, any nodes which locate at a segment $\mathbf{D_{k}}$ will receive collisions from other nodes which send packets in the slots of $\mathbf{C_{k}}$. Let $\mathbf{D}=\{\mathbf{D_{1}},\mathbf{D_{2}},\cdots, \mathbf{D_{K}}\}$ be the locations of sending nodes, and their possible interference slots are $\mathbf{C}=\{\mathbf{C_{1}},\mathbf{C_{2}},\cdots, \mathbf{C_{K}}\}$, where $\mathbf{C_{k}} {\subseteq}\mathbf{M}$, for $k\in \{1,2, \cdots, K\}$ according to Proposition \ref{proposition1}. 

According to \eqref{equa-possible-interfering-region}, an $\text{IR}_i^{\Delta m}$ is valid only when it satisfies the following conditions
\begin{gather}
d_i+\Delta m\cdot t_{slot}+t_f\geq 0, \Delta m<0, \label{eq:critical1}\\
d_i+\Delta m\cdot t_{slot}-t_f\leq R, \Delta m\geq 0.\label{eq:critical2}
\end{gather}
Now we can have the critical points to divide $\mathbf{D_{k}}$ as
\begin{equation}
X^{\Delta m}=\left\{\begin{matrix}
  -\Delta mt_{slot}-t_{f},&\Delta m<0 \\
  R-\Delta mt_{slot}+t_{f},&\Delta m > 0
\end{matrix}\right.
\end{equation}
Arranging $\{0, X^{\Delta m}, R  |  \Delta m\in  \mathbf{M}\}$ in an ascending order as $\{ X_k | k\in \{1,2, \cdots, K+1\}\}$, we have $\mathbf{D_{k}}$ as 
\begin{equation}\label{eq:Dk}
\mathbf{D_{k}}=[X_k, X_{k+1}]_0^R.
\end{equation}
The corresponding $\Delta m$ meeting \eqref{eq:critical1} and \eqref{eq:critical2} should belong to $\mathbf{C_{k}}$. We use Fig. \ref{fig_Dk_example} to exemplify the division of $\mathbf{D_{k}}$ and $\mathbf{C_{k}}$.

It is worth noting that the IRs may not cover the whole transmission range. In this sense, the nodes outside IRs (e.g., the white space in Fig. \ref{fig_horizontal_conflicit_region}) can reuse the channel with the tagged node $i$.  Concurrent transmissions are possible in UANs. 

Proposition \ref{proposition2} has revealed the IRs for a given slot. We next look at the spatial reuse opportunities.

\begin{proposition} 
When the slot length setting satisfies $t_{slot}>2t_f$, there exist \emph{collision-free regions} (CFRs) where a packet transmission at any slot will always not collide with an ongoing transmission. 
\label{proposition3}
\end{proposition}

\begin{proof}
We use Fig. \ref{fig-interfering-segement} to prove Proposition \ref{proposition3}. Proposition \ref{proposition2} has identified the IRs for a given node $i$ which locates at $d_i$ away from the sink. Fig. \ref{fig-segement-demosration-small-pktduration} and Fig. \ref{fig-segement-demosration-big-pktduration} show the IRs of different sending slots. We can see that when $t_{slot}>2t_f$, the IRs of any slots will not overlap with each other and there is still some blank space between these IRs. In other words, any node located at the blank space can send its packets at any slot without collision with node $i$. 
\IEEEQEDhere
\end{proof}

The opportunities of spatial reuse in space-time coupling underwater acoustic channel are revealed by Proposition \ref{proposition3}. The transmission in the blank space in Fig. \ref{fig-segement-demosration-small-pktduration} will not cause neither intra-slot collision nor inter-slot collision with the tagged node $i$. In this sense, two concurrent transmissions are allowed to reuse the channel. However, when $t_{slot}\leq 2t_f$ satisfies\footnote{we can increase the packet length or shorten the slot length to achieve this condition.}, the blank space in Fig. \ref{fig-segement-demosration-small-pktduration} and spatial reuse will disappear. In this case, IRs of adjacent slots will coincide with each other, and we call the coincided IR as a \emph{deep interference region} (DIR).

It is worth noting that only the IRs of two adjacent slots would coincide with each other due to the fact that $t_f<t_{slot}$.

%

\begin{proposition} 
DIR only exists between IRs of two consecutive slots.
\label{proposition4}
\end{proposition}
\begin{proof}
We consider three consecutive slots, which are slot $\Delta m$, slot $\Delta m+1$, and slot $\Delta m+2$.
According to \eqref{equa-possible-interfering-region}, $\text{IR}_i^{\Delta m}$ and $\text{IR}_i^{\Delta m+2}$ are $[d_{i}+v(\Delta m\cdot t_{slot}-t_{f}), d_{i}+v(\Delta m\cdot t_{slot}+t_{f})]$ and $[d_{i}+v( (\Delta m+2)\cdot t_{slot}-t_{f}), d_{i}+v( (
\Delta m+2)\cdot t_{slot}+t_{f})].$
Since $t_{f}\leq t_{slot}$ always holds, we have $d_{i}+v(\Delta m\cdot t_{slot}+t_{f})\leq d_{i}+v( (\Delta m+2)\cdot t_{slot}-t_{f})$. Thus, $\text{IR}_i^{\Delta m}$ and $\text{IR}_i^{\Delta m+2}$ are impossible overlapped. The DIR could be the overlapped IR of only two consecutive slots. 
\IEEEQEDhere
\end{proof}

\section{Performance Analysis for Slotted MAC}\label{sect:perf}
We have known that a long slot length (i.e., $\beta\geq1$) can eliminate inter-slot collisions (c.f., Proposition \ref{proposition1}) and accommodate channel-reuse opportunities (c.f., Proposition \ref{proposition3}). However, a long slot length will lead to low channel utilization and severe intra-slot collisions due to the suppressed backlogs in the long slot. In this section, we further study the collision probabilities and the transmission throughput with respect to the slot length, packet duration, packet arrival rate, and the number of nodes.


\subsection{Successful Transmission Probability and Throughput}\label{sect:psNT}
Based on the discussion in Section \ref{sect:collisions}, we assume that a tagged node $i$ ($i\in \mathbf{N}$) who locates at $d_{i}\in\mathbf{D_{k}}$ sends a packet in  slot $\Delta m=0$ (i.e., $t_{i}=0$). The collisions may come from the sending at $t_{j}=\Delta m\cdot t_{slot}$, which is in the slots $\Delta m \in \mathbf{C_{k}}$. 

We use the typical Slotted-ALOHA to study the performance of slotted MAC.  
 A node using Slotted-ALOHA sends packets at the beginning of a slot once a packet arrives at the MAC layer before this slot. To avoid collisions with node $i$, all the nodes in $\text{IR}_i^{\Delta m}$ ($\forall\Delta m \in \mathbf{C_{k}}$) should keep silent in the corresponding slots, i.e., no packets arrive at these nodes in the previous slot $\Delta m-1$. 

 Assume that the packet arrivals follow a Poisson process having an arrival rate of $\lambda$. The probability that no packets arrive in a slot can be calculated by $e^{-\lambda\cdot t_{slot}}$. 
The transmission probability for each slot is then 
\begin{equation}
 p_{t} = 1-e^{-\lambda\cdot t_{slot}}.
  \label{equa-transmission-probability}
\end{equation}


Based on Proposition \ref{proposition3}, we study the probability of successful transmissions in two cases, i.e., $t_{slot}>2t_f$ and $t_f\leq t_{slot}\leq 2t_f$. 

\begin{enumerate}
 \item Case I of $t_{slot}> 2t_f$: 
Let $P_{z,k}$ denote the probability that a node locates in the IRs. As shown in Fig. \ref{fig-segement-demosration-small-pktduration}, the IRs of the interference slots in $\mathbf{C_{k}}$ do not coincide with each other. Then, we have $P_{z,k}=\sum_{\Delta m \in \mathbf{C_{k}}}\!p_{z,k}^{\Delta m} $, where $p_{z,k}^{\Delta m}$ is the probability that a node is inside $\text{IR}_i^{\Delta m}$. Then $1-P_{z,k}$ is the probability that a node falls in CRFs in Case I. Nodes in IRs should keep silent to avoid colliding with nodes in CFRs. 
Thus the non-collision probability $\Pr \{\textbf{NC}|d_{i} \in \mathbf{D_{k}}\}$ conditioned on node $i$ $\mathbf{D_{k}}$ can be expressed as follow 
  \begin{equation}
  \begin{array}{l}
   \Pr \{\textbf{NC}|d_{i} \in \mathbf{D_{k}}\}  \\ 
= \sum\limits_{j = 0}^{N-1} {\left( {\begin{array}{*{20}{c}}
     {N - 1}  \\
     j  \\
  \end{array}} \right){{\left( {1 - {P_{z,k}}} \right)}^{N - 1 - j}}} {\left( {{P_{z,k}}} \right)^j}{\left( {1 - {p_t}} \right)^j} \\ 
= {(1 - p_{t}\! \cdot\! {P_{z,k}})^{N - 1}}. \\ 
   \end{array}
   \label{equa-segemnt-ps-smalltf}
  \end{equation}
 \item Case II of $t_f\leq t_{slot}\leq 2t_f$:  The collisions in Case II are much more complicated than in Case I. Fig. \ref{fig-segement-demosration-big-pktduration} shows that the IRs of different slots coincide with each other, and Proposition \ref{proposition4}  shows that only IRs of two consecutive slots have a coincided region. Therefore, the event $\textbf{NC}$ is equivalent to that the nodes inside the DIRs have no packet arrivals in $2\!\cdot\!t_{slot}$ duration and the other nodes outside the DIRs have no packet arrivals in $t_{slot}$ duration.
 Let $P_{o,k}$ denote the probability that a node inside the DIR conditioned on $d_{i}\in \mathbf{D_{k}}$. 
 Similarly, we have $P_{o,k}=\sum_{\Delta m \in \mathbf{C_{k}}}\!p_{o,k}^{\Delta m, \Delta m+1}$, where $p_{o,k}^{\Delta m, \Delta m+1} $ is the probability that a node is inside the DIR of two consecutive $\text{IR}_i^{\Delta m}$ and $\text{IR}_i^{\Delta m+1}$. The probability $\Pr \{\textbf{NC}|d_{i} \in \mathbf{D_{k}}\}$ in Case II can be expressed as follow
  \begin{equation}
  \begin{array}{l}
   \Pr \{\textbf{NC}|d_{i} \in \mathbf{D_{k}}\}   \\ 
 = \sum\limits_{j = 0}^{N-1} {\left( {\begin{array}{*{20}{c}}
     {N - 1}  \\
     j  \\
   \end{array}} \right){{\left( {1 - {P_{o,k}}} \right)}^{ j}}}{\left( {1 - {p_t}} \right)^j} \\
   \quad\quad\cdot{\left( {{P_{o,k}}} \right)^{N-1-j}}{\left({1 - {p_{c}}}\right)^{\left(N-1-j\right)} }\\ 
  = {\left((1 - p_{t})\cdot \left(1-{P_{o,k}}+(1-p_{t})\cdot{P_{o,k}} \right)\right)^{N - 1}},\\ 
   \end{array}
   \label{equa-segemnt-ps-bigtf_2}
\end{equation}
where $1- p_{c}$ is the probability that a node keeps silent for two slots and $p_{c}=1-e^{-\lambda\cdot 2t_{slot}}$. The last equation in \eqref{equa-segemnt-ps-bigtf_2} follows due to $(1-p_{c})=(1-p_{t})^2$.
\end{enumerate}

It is interesting to point out that, we can have the following relations from Fig. \ref{fig-interfering-segement}. 
\begin{itemize}
\item $P_{z,k}(t_{f}=a, \beta\tau=b)=1-P_{o,k}(t_{f}=b, \beta\tau=a)$ in Case I: It can be observed from Fig. \ref{fig-interfering-segement} that the widths of a CFR and a DIR are the same as $v|t_f-\beta\tau|$. As long as $d_i$ keeps unchanged, the areas of CFRs and DIRs will not change. Then, we have $P_{z,k}(t_{f}=a, \beta\tau=b)=1-(1-P_{z,k}(t_{f}=a, \beta\tau=b))=1-P_{o,k}(t_{f}=b, \beta\tau=a)$.
\item $P_{z,k}=1+P_{o,k}$ in Case II: Two consecutive IRs overlaps in Case II. The overlapped areas are then DIRs. Therefore, we have $P_{z,k}=\sum_{\Delta m \in \mathbf{C_{k}}}\!p_{z,k}^{\Delta m} =1+P_{o,k}$.
\end{itemize}


The probability $\Pr \{\textbf{NC}|d_{i} \in \mathbf{D_{k}}\}$ gives only the successful transmission probability of a tagged node $i$ who locates in $\mathbf{D_{k}}$. We further consider the distribution of node $i$ to obtain the successful transmission probability $P_{s}$. Let $P_{k}$ denote the probability of $d_{i}\in \mathbf{D}_{k}$, and let $P_{s,k}=\Pr \{\textbf{NC}|d_{i} \in \mathbf{D_{k}}\}$.  The successful transmission probability $P_{s}$ can be expressed by 
\begin{equation}
 P_{s} = \sum_{\mathbf{D_{k}}\in \mathbf{D}} P_{k}\cdot P_{s,k}. 
  \label{equa-total-successful-rate}
\end{equation}

Based on the successful transmission probability and packet arrival rate, we have the throughput $T$ of the slotted MAC in UANs as follows: 
\begin{equation}
T
=\frac{N\!\cdot\! \lambda\! \cdot\! t_{f}\!\cdot\!t_{slot}}{t_{slot}}\cdot\! P_{s}
=N\!\cdot\! \lambda\! \cdot \!t_{f}\!\cdot\! P_{s}.
\label{equa_throughput}
\end{equation}

The throughput in \eqref{equa_throughput} is an aggregate network throughput for all the nodes and is normalized by the packet duration. 
\begin{proposition} 
The successful transmission probabilities and throughput are the same at $\beta=0$ and at $\beta\tau=t_f$ under Poison packet arrivals. 
\label{proposition5}
\end{proposition}
\begin{proof}
When $\beta=0$, there is no guard interval and all the DIRs cover the whole communication range of the receiver, As shown in Fig. \ref{fig-interfering-segement}. Then, we have  have $P_{o,k}=1$. Based on \eqref{equa-transmission-probability} and \eqref{equa-segemnt-ps-bigtf_2}, we have $p_t=1- e^{-\lambda\cdot t_{f}}$ and $\Pr \{\textbf{NC}|d_{i} \in \mathbf{D_{k}}\} = (1-p_{t})^{2(N-1)}$, respectively. Therefore, the non-collision probability for a tagged node becomes $P_{s,k}=\Pr \{\textbf{NC}|d_{i} \in \mathbf{D_{k}}\} = e^{-2(N-1)\lambda\cdot t_{f}}$. 

When $\beta\tau=t_{f}$, there are neither CFR nor DIR and all the IRs cover the whole communication range of the receiver. We have $P_{z,k}=1$. We also have  $p_t=1- e^{-2\lambda\cdot t_{f}}$ and $\Pr \{\textbf{NC}|d_{i} \in \mathbf{D_{k}}\} = (1-p_{t})^{N-1}$, respectively. We then obtain $P_{s,k}=\Pr \{\textbf{NC}|d_{i} \in \mathbf{D_{k}}\} = e^{-2(N-1)\lambda\cdot t_{f}}$. 

It is observed that $P_{s,k}$ is independent with the location of node $i$ when $\beta=0$ and $\beta\tau=t_f$. According to \eqref{equa-total-successful-rate} and \eqref{equa_throughput}, the successful transmission probabilities and throughput at $\beta=0$ and $\beta\tau=t_f$ are the same. 
\IEEEQEDhere
\end{proof}

We will further show the existence of the maximizer for the slot length to reach a peak performance in uniformly distributed UANs.


\subsection{Case Study for Uniformly Distributed UANs}\label{sect:case}

The node distribution in UANs determines the probabilities $P_{k}$, $P_{z,k}$, and $P_{o,k}$ in the discussion of Section \ref{sect:psNT}. In this section, we try to calculate these probabilities in a UAN where nodes are uniformly distributed, and then derive the successful transmission probability and throughput using \eqref{equa-total-successful-rate} and \eqref{equa_throughput}.

In a uniformly distributed UAN, the probability that any node locates inside a region is proportional to the area of that region. Since the collision happens only at the sink, we focus on the coverage area of the sink. Let $A$ denote the area of the sink's coverage and $S_{k}$ denote the area of $\mathbf{D_{k}}$'s coverage. Thus, a sender falls in $\mathbf{D_{k}}$ with the following probability
\begin{equation}
 P_{k}=\frac{S_{k}}{A}.
\label{equa_PositionProbability_Pk}
\end{equation}

We next consider the different underwater acoustic propagation features in vertical and horizontal directions, which will lead to \emph{transmission range inconsistency}. The vertical transmission faces much less multi-path interferences caused by the sound reflection and refraction than the horizontal transmission \cite{Ref-Zhong}. For the same range of transmissions, the vertical transmission obtains a better signal-to-interference and noise ratio (SINR) than the horizontal transmission. Thus, the vertical transmission can reach a longer range than the horizontal transmission, resulting in a disk-shaped coverage for horizontal transmissions and an ellipse-shaped coverage for vertical transmissions, as shown in Fig. \ref{fig_voceraget_region}.


\begin{figure}
	\centering
	\subfloat[]{
	\includegraphics[width=0.45\columnwidth]{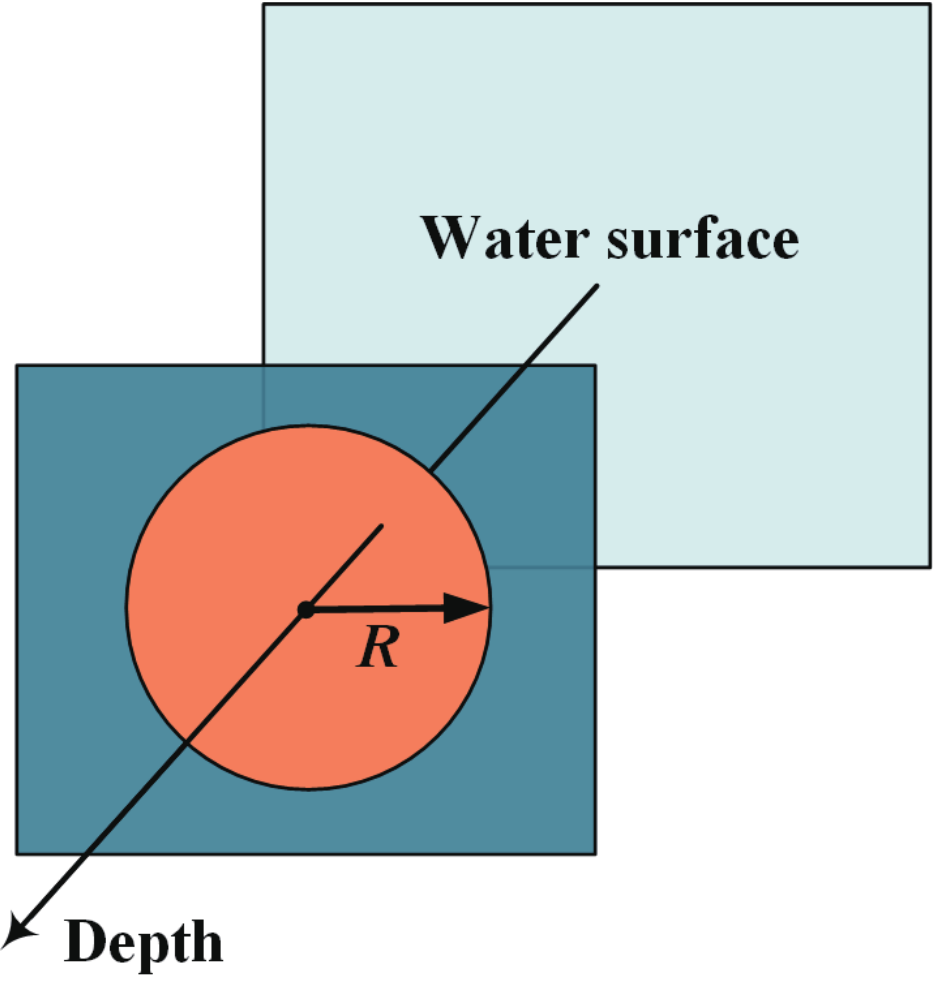}
	\label{fig-horizontal-coverage}}
	\hfill
	\subfloat[]{
	\includegraphics[width=0.45\columnwidth]{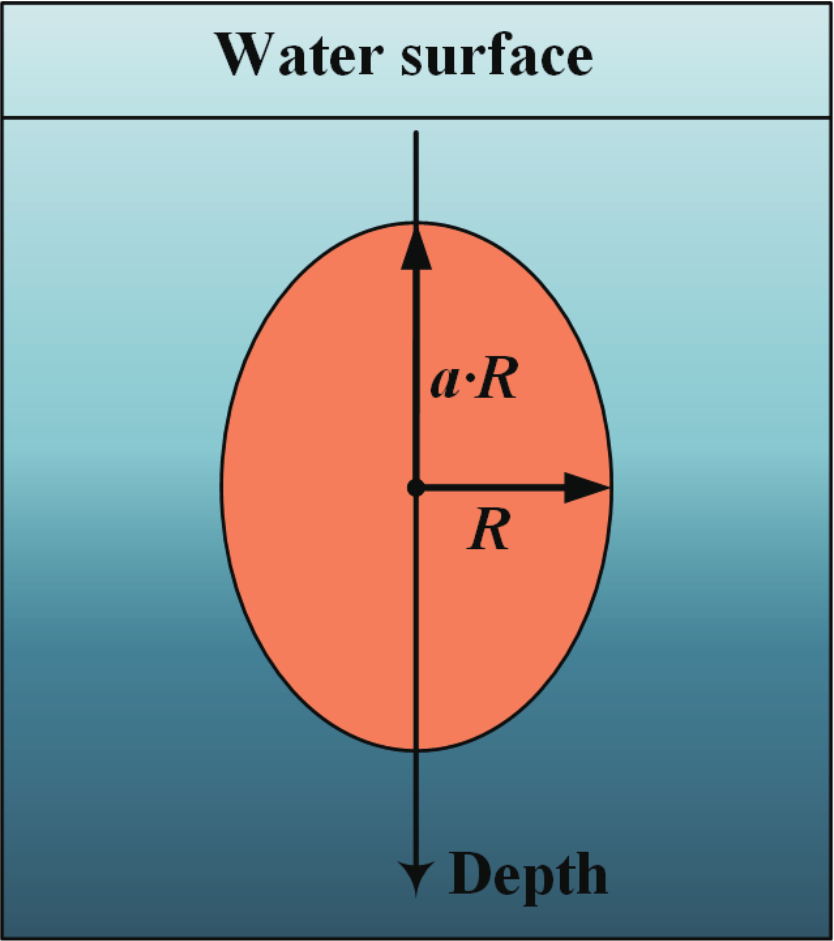}
	\label{fig-vertical-coverage}}
	\caption{Transmission range inconsistency in UANs. The range is (a) a circle for horizontal transmissions, and (b) an ellipse for vertical transmissions. }
	\label{fig_voceraget_region}
\end{figure}

To consider the transmission range inconsistency, 
let the radius of the disk-shaped horizontal coverage be $R$, and suppose the ellipse-shaped vertical coverage has a minor axis of $R$ and a major axis of $\alpha R$, where $\alpha>1$. Then, we have $A=\pi R^2$ in the horizontal transmission plane and $A=\pi\alpha R^2$ in the vertical transmission plane, respectively. The maximum propagation delay for vertical transmissions becomes $\alpha\tau$. Different from the IRs in the horizontal transmission plane in Fig. \ref{fig_horizontal_conflicit_region}, there are totally 8 different IRs  in the vertical transmission plane, as shown in Fig. \ref{fig_vertical_conflicit_region}.  

%


\begin{figure}
	\centering
	\subfloat[]{
	\includegraphics[width=0.22\columnwidth]{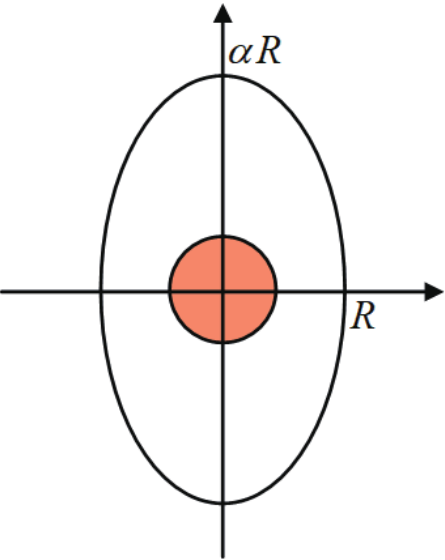}
	\label{fig-vertical-a}}
	\hfill
	\subfloat[]{
	\includegraphics[width=0.22\columnwidth]{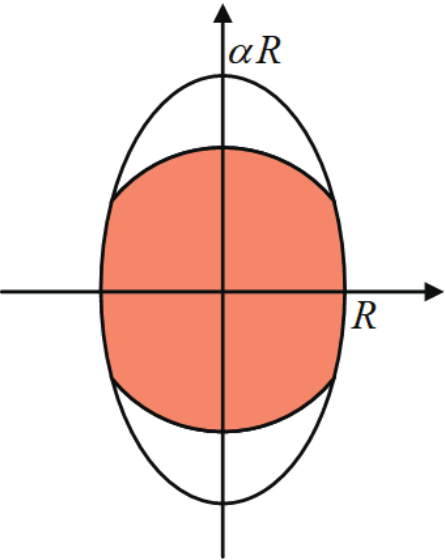}
	\label{fig-vertical-b}}
	\hfill
	\subfloat[]{
	\includegraphics[width=0.22\columnwidth]{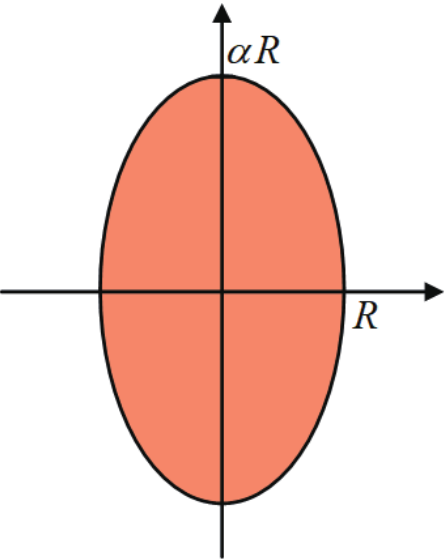}
	\label{fig-vertical-c}}
	\hfill
	\subfloat[]{
	\includegraphics[width=0.22\columnwidth]{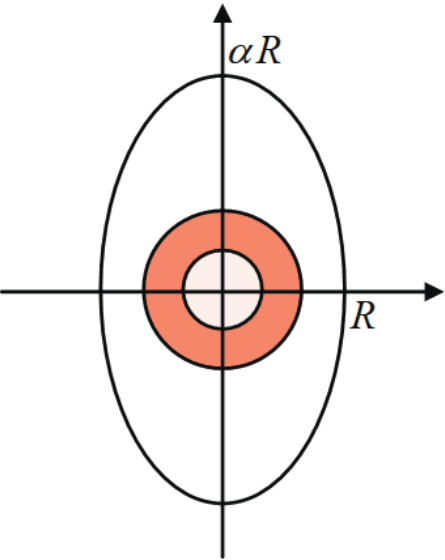}
	\label{fig-vertical-d}}\\

	\subfloat[]{
	\includegraphics[width=0.22\columnwidth]{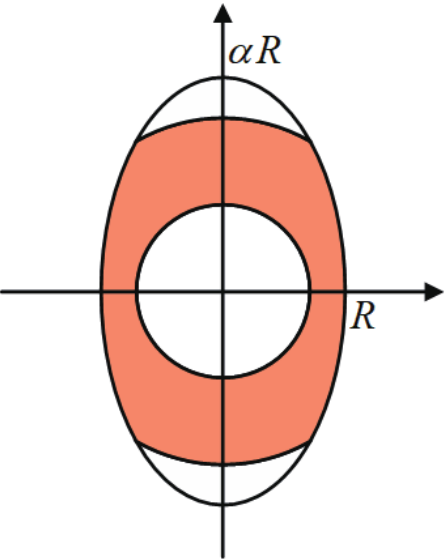}
	\label{fig-vertical-e}}
	\hfill
	\subfloat[]{
	\includegraphics[width=0.22\columnwidth]{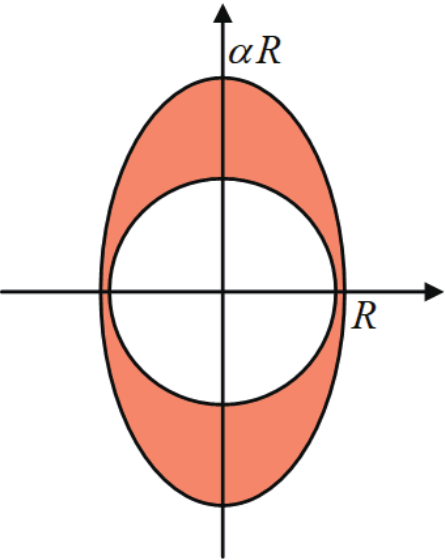}
	\label{fig-vertical-f}}
	\hfill
	\subfloat[]{
	\includegraphics[width=0.22\columnwidth]{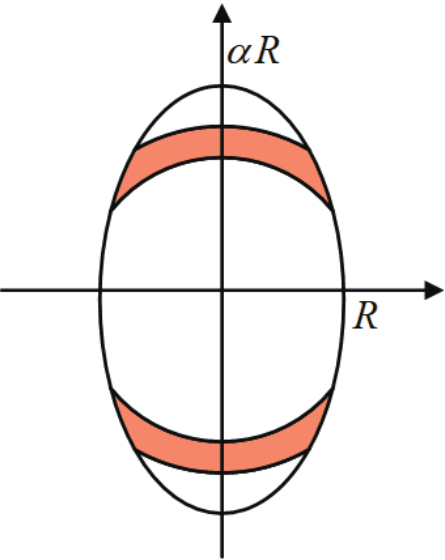}
	\label{fig-vertical-g}}
	\hfill
	\subfloat[]{
	\includegraphics[width=0.22\columnwidth]{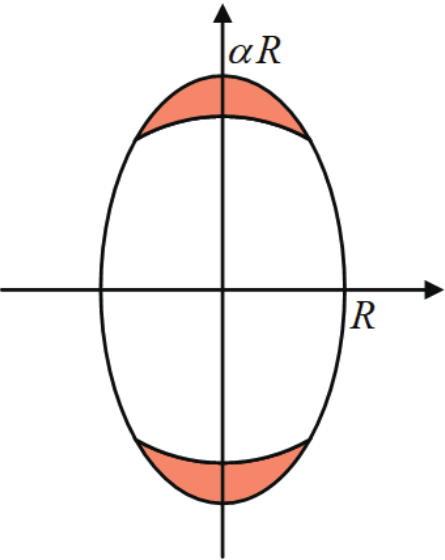}
	\label{fig-vertical-h}}\\	

	\caption{Eight possible IRs in the vertical transmission plane.}
	\label{fig_vertical_conflicit_region}
\end{figure}

For any sending node, let $S^{\Delta m}_{1,k}$ denote the area of its $\text{IR}_i^{\Delta m}$ and let $S^{\Delta m}_{2,k}$ denote the area of DIR for slots $\Delta m$ and $\Delta m+1$. The non-trivial expressions for $S^{\Delta m}_{1,k}$ and $S^{\Delta m}_{2,k}$ in the eight IRs in Fig. \ref{fig_vertical_conflicit_region} are derived in Appendix \ref{appendix:S}.

The probabilities $P_{z,k}$ and $P_{o,k}$ can be calculated as follows    
\begin{gather}
 P_{z,k} = \sum_{\Delta m\in\mathbf{C_{k}}}\frac{\mathbb{E}[S^{\Delta m}_{1,k}]}{\int_{d^{k}}f(l)Adl},
 \label{equa_PositionProbability_pzk}\\
 P_{o,k} = \sum_{\Delta m\in\mathbf{C_{k}}}\frac{\mathbb{E}[S^{\Delta m}_{2,k}]}{\int_{d^{k}}f(l)Adl},
\label{equa_PositionProbability_pok}
\end{gather}
where
\begin{equation}
 \mathbb{E}[S^{\Delta m}_{1,k}] = \int_{\mathbf{D_{k}}}f(l)S^{\Delta m}_{1,k}dl,\quad \mathbb{E}[S^{\Delta m}_{2,k}] = \int_{\mathbf{D_{k}}}f(l)S^{\Delta m}_{2,k}dl,
\end{equation}
where $f(l)$ is the link distance distribution function in UANs, and $\mathbb{E}[\cdot]$ is the expectation.
In the horizontal transmission plane, we have the link distance's probability density function (PDF) in a uniformly distributed disk as follows \cite{Ref-RandomNetwork} 
\begin{equation}
f(l) \!=\! \frac{{2l}}{{{R^2}}}\left( {\frac{2}{\pi }{{\cos }^{ - 1}}\left( {\frac{l}{{2R}}} \right) - \frac{l}{{\pi R}}\sqrt {1 - \frac{{{l^2}}}{{4{R^2}}}} } \right),0\! < l \!< 2R.
\label{equa_linkPDF_disk}
\end{equation}
In the vertical transmission plane, 
the PDF of link distance $l$ could be expressed as \cite{miller-distribution}
\begin{equation}
f(l) = \frac{9l}{2\alpha}\cdot \exp(-\frac{9l^{2}(\alpha^2+1)}{8\alpha^2})\cdot \mathrm{I_{0}}(\frac{9l^{2}(\alpha^2-1)}{8\alpha^2}), 
\label{equa_rectangle_pdf}
\end{equation}  
where the $\mathrm{I_{0}}$ is the modified Bessel function of the first kind. 

Substituting \eqref{equa_PositionProbability_Pk}, \eqref{equa_PositionProbability_pzk}, and \eqref{equa_PositionProbability_pok} to \eqref{equa-total-successful-rate} and \eqref{equa_throughput}, we can obtain the successful transmission probability and throughput for the uniformly distributed UANs. 

\begin{proposition} 
There exists a maximizer of the slot length between $(t_f, 2t_f)$ to achieve a peak successful transmission probability and throughput in uniformly distributed UANs.
\label{proposition7}
\end{proposition}
\begin{proof}
According to the Proposition \ref{proposition5}, we only have to prove there is at least one slot length in $(t_f, 2t_f)$ to achieve better successful transmission probability and throughput than the slot length $t_f$ or $2t_f$. 
Suppose $\beta\tau=a$ ($a\in (0, t_f$)), we need to prove
\begin{equation}
 P_{s}(\beta\tau=a)\geq P_{s}(\beta\tau=t_f). \label{extremevalue-condition}
\end{equation} 

Based on \eqref{equa-transmission-probability}, \eqref{equa-segemnt-ps-bigtf_2} and \eqref{equa-total-successful-rate}, we have
\begin{equation}
\begin{aligned}
 \!P_{s} 
     \!  &= \sum_{\mathbf{D_{k}}\in \mathbf{D}}  \! P_{k} \!\left((e^{-\lambda t_{slot}} \!) (1 \!-\! {P_{o,k}}) \!+\!(e^{-2\lambda t_{slot}}){P_{o,k}} \right)^{\!N \!-\! 1}.
\end{aligned}
\end{equation}
Since the exponential function is a convex function, using Jensen's inequality, we have
\begin{equation}
\begin{aligned} 
 P_{s} &\geq\sum_{\mathbf{D_{k}}\in \mathbf{D}}P_{k}e^{-\lambda(N - 1)t_{slot}(1+P_{o,k})}\\
       &\geq e^{-\lambda(N - 1)t_{slot}\sum_{\mathbf{D_{k}}}P_{k}(1+P_{o,k})},
\end{aligned}
\end{equation}   
We have derived $P_{s}(\beta\tau=t_f)=e^{-2(N-1)\lambda\cdot t_{f}}$ in Proposition \ref{proposition5}. Then, \eqref{extremevalue-condition} is equivalent to 
\begin{equation}
 \sum_{\mathbf{D_{k}}}P_{k}(P_{o,k}(\beta\tau=a))\leq \frac{t_f-a}{t_f+a}.
 \label{equa_jensen_gap1}
\end{equation}
Define $\Psi$ as follow,s 
\begin{equation}
 \Psi=\frac{t_f-a}{t_f+a}-\sum_{\mathbf{D_{k}}}P_{k}(P_{o,k}(\beta\tau=a)).
\end{equation}

The probabilities $P_k$ and $P_{o,k}$ depend on the node distribution. Due to the non-trivial  expressions of $P_k$ and $P_{o,k}$, we use the numerical results of $\Psi$ with respect to $t_f$ and $\beta\tau$ in Fig. \ref{fig_Proof_vertical_gap} to show that $\Psi>0$ is satisfied for $0<\beta\tau<t_f$. Thus, \eqref{extremevalue-condition} satisfies and Proposition \ref{proposition7} is proved. 
\begin{figure}
\centering
\includegraphics[width=0.8\columnwidth]{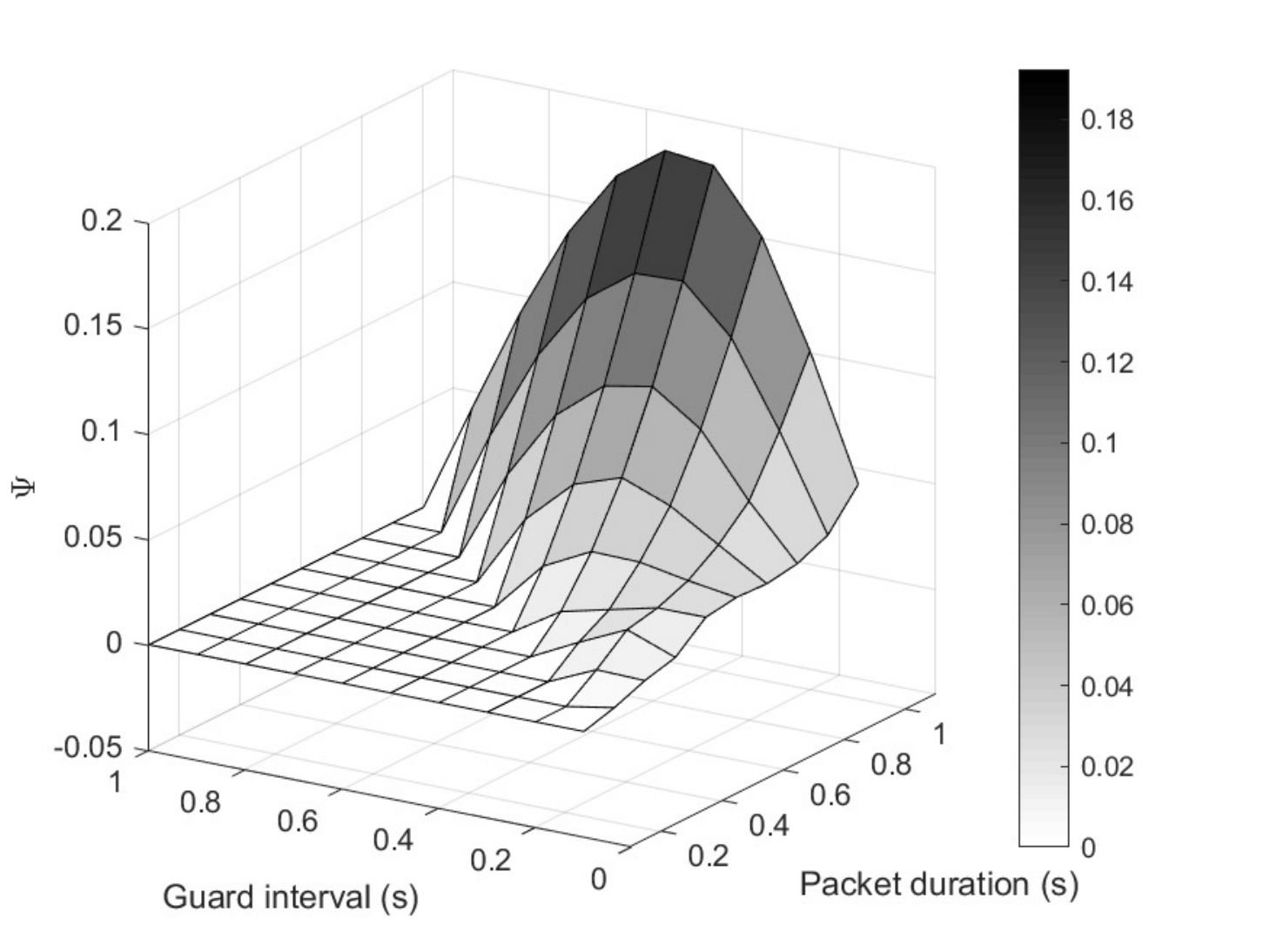}
\caption{$\Psi>0$ is always satisfied for $\beta\tau<t_f$.}
\label{fig_Proof_vertical_gap}
\end{figure}
\IEEEQEDhere
\end{proof}

\begin{proposition} \label{proposition6}
The vertical transmissions achieve the same successful transmission probability and throughput as the horizontal transmissions when $t_{slot}=2t_f$.
\end{proposition}
\begin{proof}
The proof for Proposition \ref{proposition5} has obtained $P_{z,k}=1$ when $t_{slot}=2t_f$ for both horizontal transmissions and vertical transmissions, since there are neither CFR nor DIR in this case. Therefore, the successful transmission probabilities and throughput are the same.
\IEEEQEDhere
\end{proof}

\section{Simulation Results and Discussions}\label{sect:sim}
This section uses  the simulation results in NS3 to verify the analytical results in Section \ref{sect:collisions} and Section \ref{sect:perf}, and to study the performance of the space-time coupling MAC in UANs. 

In the simulation, we use UniformRandomVariable to generate the coordinates of nodes.
The maximal propagation range of horizontal transmission is 1500 meters (i.e., one-second propagation time at the speed of 1500 m/s), the data rate is 4000 bps. The the packet size determines the packet duration, e.g., when the packet size is 500 bytes, the transmission duration for a packet is 1 second at 4000 bps.


\subsection{Verification of $P_{k}$, $P_{z,k}$, $P_{o,k}$ and $P_{s,k}$}
We first verify the probabilities of $P_{k}$, $P_{z,k}$, $P_{o,k}$ and $P_{s,k}$. In the simulations, two nodes are randomly deployed in the coverage area of the receiver. The first node is taken as the sending node, and the other is the interference node. We carry out 100,000 runs of simulations to obtain the probabilities. Both cases of CFR and DIR are considered using $\beta=0.4$.


\subsubsection{CFR case} 
In the simulation,  $t_{f}$ is set to $0.1\tau$, which is smaller than the guard duration, i.e., $t_f<\beta\tau$. Then, $M$ is 2 and $K$ is 5. According to the Proposition \ref{proposition3}, the CFR exists.

	\begin{table}[!ht]
	  \caption{$P_{k}$, $P_{z,k}$ and $P_{s,k}$ in the CFR case.}
	  \centering
	  \setlength{\tabcolsep}{1.5mm}{
	  \begin{tabular}{|c|c|c|c|c|c|c|} 
	    \hline
	    \multicolumn{7}{|c|}{\textbf{Distance segment index}}\\  
	    \hline
	    \multicolumn{2}{|c|}{$k$}            & 1       & 2         & 3         & 4         & 5  \\
	    \hline
	    \multicolumn{2}{|c|}{$\mathbf{D_{k}}$}        & [0,0.1] & [0.1,0.4] & [0.4,0.6] & [0.6,0.9] & [0.9,1]\\
	    \hline
	    \multicolumn{7}{|c|}{\textbf{$P_{k}$ Validation}}\\
	    \hline
	    \multicolumn{2}{|c|}{Theo $P_{k}$}   & 0.01   & 0.15   & 0.2    & 0.45   & 0.19 \\    
	    \hline
	    \hline
	    \multicolumn{2}{|c|}{Sim  $P_{k}$}   & 0.0097 & 0.1494 &0.1992  & 0.4522 & 0.1895 \\
	    \hline
	    & \textbf{IR of} &\multicolumn{5}{|c|}{\textbf{$P_{z,k}$ Validation}}\\
	    \hline   
	    {Theo $P_{z,k}^{-2}$} & $\Delta m = -2$     & $\backslash$ & $\backslash$ &  $\backslash$ &  $\backslash$ & 0.0033 \\  
	    \hline   
	    {Theo $P_{z,k}^{-1}$} & $\Delta m = -1$     & $\backslash$ & $\backslash$ & 0.0141 & 0.1007 & 0.1799 \\ 
	    \hline   
	    {Theo $P_{z,k}^{0}$} & $\Delta m =  0$     & 0.0282  & 0.11   & 0.2015 & 0.3007 & 0.2771 \\ 
	    \hline   
	    {Theo $P_{z,k}^{1}$} & $\Delta m=  1$     & 0.2265  & 0.31   & 0.18   & $\backslash$ &  $\backslash$ \\ 
	    \hline   
	    {Theo $P_{z,k}^{2}$} & $\Delta m =  2$     & 0.0657  & $\backslash$ & $\backslash$ & $\backslash$ & $\backslash$ \\ 
	    \hline
	    \hline   
	    {Sim $P_{z,k}^{-2}$} & $\Delta m = -2$     & $\backslash$ & $\backslash$ & $\backslash$ &  $\backslash$ & 0.0036 \\  
	    \hline   
	    {Sim $P_{z,k}^{-1}$} & $\Delta m = -1$     &  $\backslash$ & $\backslash$ & 0.0155 & 0.1043 & 0.1806 \\ 
	    \hline   
	    {Sim $P_{z,k}^{0}$} & $\Delta m =  0$     & 0.0339  & 0.1139 & 0.2013 & 0.3050 & 0.2722 \\ 
	    \hline   
	    {Sim $P_{z,k}^{1}$} & $\Delta m =  1$     & 0.2148  & 0.3104 & 0.171  &  $\backslash$ & $\backslash$  \\ 
	    \hline   
	    {Sim $P_{k}^{2}$} & $\Delta m =  2$     & 0.0606  &  $\backslash$ &  $\backslash$ &  $\backslash$ &  $\backslash$\\    
	    \hline
	    \multicolumn{7}{|c|}{\textbf{$P_{s,k}$ Validation}}\\
	    \hline
	    \multicolumn{2}{|c|}{Theo $P_{s,k}$}   & 0.9291  & 0.9071 & 0.9125 & 0.9112 & 0.8982 \\    
	    \hline
	    \hline
	    \multicolumn{2}{|c|}{Sim  $P_{s,k}$}   & 0.9316  & 0.9071 & 0.9142 & 0.9106 & 0.8981 \\
	    \hline
	  \end{tabular}}
	  \label{table-distancesegment-case1}
	\end{table}

The theoretical and simulation results are summarized in Tab. \ref{table-distancesegment-case1}. It is shown that the values of $P_{k}$, $P_{z,k}$ and $P_{s,k}$ in the simulation are very close to the theoretic values, which verifies our analysis in Section \ref{sect:case}.

For any given region that the sending node falls in, the corresponding interference slots verify the conclusion in Proposition \ref{proposition3}. For example, when the sending node is in the segment $\mathbf{D_{1}}$, since $M=2$, only the current slot and the next two slots will collide with its current transmission, and the collision probabilities in other slots are zero. Note that the sum of collision probabilities over all the slots (i.e., $P_{z,k}=\sum P^{\Delta m}_{z,k}$)  is less than one. This indicates that there exist some regions where the other nodes can send simultaneously in any slot without collisions, which verifies the existence of CFRs. 

	
  \subsubsection{DIR case} 
 The packet duration $t_{f}$ is set to 0.7$\tau$, which is greater than the guard duration, to generate the DIR case. In this case, $M$ is 1 and $K$ is 5. The theoretical and simulation results are summarized in Tab. \ref{table-distancesegment-case2}.
	\begin{table}[!ht]
	  \caption{$P_{k}$, $P_{z,k}$, $P_{o,k}$ and $P_{s,k}$ in the DIR case.}
	  \centering
	  \setlength{\tabcolsep}{0.55mm}{
	  \begin{tabular}{|c|c|c|c|c|c|c|} 
	    \hline
	    \multicolumn{7}{|c|}{\textbf{Distance segment index}}\\  
	    \hline
	    \multicolumn{2}{|c|}{$k$}            & 1       & 2         & 3         & 4         & 5  \\
	    \hline
	    \multicolumn{2}{|c|}{$\mathbf{D_{k}}$}   & [0,0.3] & [0.3,0.4] & [0.4,0.6] & [0.6,0.7] & [0.7,1]\\
	    \hline
	    \multicolumn{7}{|c|}{\textbf{$P_{k}$ Validation}}\\
	    \hline
	    \multicolumn{2}{|c|}{Theo $P_{k}$}   & 0.09   & 0.07   & 0.2    & 0.13   & 0.51 \\    
	    \hline
	    \hline
	    \multicolumn{2}{|c|}{Sim  $P_{k}$}   & 0.0887 & 0.0702 &0.2027  & 0.1298 & 0.5086 \\
	    \hline
	    & \textbf{IR of} &\multicolumn{5}{|c|}{\textbf{$P_{z,k}$ Validation}}\\
	    \hline   
	    {Theo $P_{z,k}^{-1}$} & $\Delta m=-1$    & $\backslash$ & $\backslash$ & 0.0141 & 0.0636 & 0.2097 \\ 
	    \hline   
	    {Theo $P_{z,k}^{0}$} & $\Delta m=0$    & 0.8086  & 1      & 1      & 1      & 0.9701 \\ 
	    \hline   
	    {Theo $P_{z,k}^{1}$} & $\Delta m=1$    & 0.6392  & 0.4341 & 0.18   & $\backslash$ & $\backslash$ \\ 
	    \hline
	    \hline     
	    {Sim $P_{z,k}^{-1}$} & $\Delta m=-1$     & $\backslash$ & $\backslash$ & 0.0147 & 0.0639 & 0.2155 \\ 
	    \hline   
	    {Sim $P_{z,k}^{0}$} & $\Delta m=0$     & 0.8124  &1        & 1      & 1      & 0.9681 \\ 
	    \hline   
	    {Sim $P_{z,k}^{1}$} & $\Delta m=1$     & 0.6306  & 0.4255 & 0.1769  & $\backslash$ & $\backslash$ \\ 
	    \hline
	    & \textbf{DIR between} &\multicolumn{5}{|c|}{\textbf{$P_{o,k}$ Validation}}\\
	    \hline   
	    {Theo $P_{o,k}^{-1,0}$} & $\Delta m=-1, \Delta m=0$    & $\backslash$ & $\backslash$ & 0.0141 & 0.0636 & 0.17911 \\ 
	    \hline   
	    {Theo $P_{o,k}^{0,1}$} & $\Delta m=0, \Delta m=1$    & 0.4478  & 0.4341 & 0.18   & $\backslash$ & $\backslash$ \\ 

	    \hline
	    \hline     
	    {Sim $P_{o,k}^{-1,0}$} & $\Delta m=-1, \Delta m=0$    & $\backslash$ & $\backslash$ & 0.0147 & 0.0639 & 0.1836 \\ 
	    \hline   
	    {Sim $P_{o,k}^{0,1}$} & $\Delta m=0, \Delta m=1$    & 0.4430  & 0.4255 & 0.1769 & $\backslash$ & $\backslash$ \\ 

	    \hline
	    \multicolumn{7}{|c|}{\textbf{$P_{s,k}$ Validation}}\\
	    \hline
	    \multicolumn{2}{|c|}{Theo $P_{s,k}$}   & 0.4676  & 0.4710 & 0.5296 & 0.5614 & 0.5331 \\    
	    \hline
	    \hline
	    \multicolumn{2}{|c|}{Sim  $P_{s,k}$}   & 0.4693  & 0.4708 & 0.5301 & 0.5617 & 0.5322 \\
	    \hline
	  \end{tabular}}
	  \label{table-distancesegment-case2}
	\end{table}

The simulation results in Tab. \ref{table-distancesegment-case2} also verify our analytical results. Specifically, we observe that $P_{z,k} >1$ at any segment $\mathbf{D_{k}}$, indicating the existence of DIRs (see Proposition \ref{proposition4}). We also observe that the probability $p_{o,k}^{\Delta m, \Delta m+1}$ of a DIR is always less than the probability $p_{z,k}^{\Delta m}$ of an interference slot, and $ P_{z,k}-1= P_{o,k}$. That means a DIR falls within a common area of the IRs of two consecutive slots, verifying Proposition \ref{proposition4}.

\subsection{Performance Discussions}

We next study the performance of Slotted-ALOHA by the NS3 simulator. We carry out 3000 runs of simulations to obtain the successful transmission probability $P_s$ and network throughput $T$.
We consider the impact of packet duration ($t_f$), guard coefficient ($\beta$), network density ($N$), and traffic load ($\lambda$) in both horizontal and vertical coverage.
In these simulations, we set $\alpha=1.5$ to simulate the vertical underwater acoustic channel. The setting of $\alpha=1$ corresponds to the horizontal underwater acoustic channel. 


\begin{figure*}[t]
\centering
\subfloat[Successful transmission probability versus packet duration.]{\includegraphics[width=3.3in]{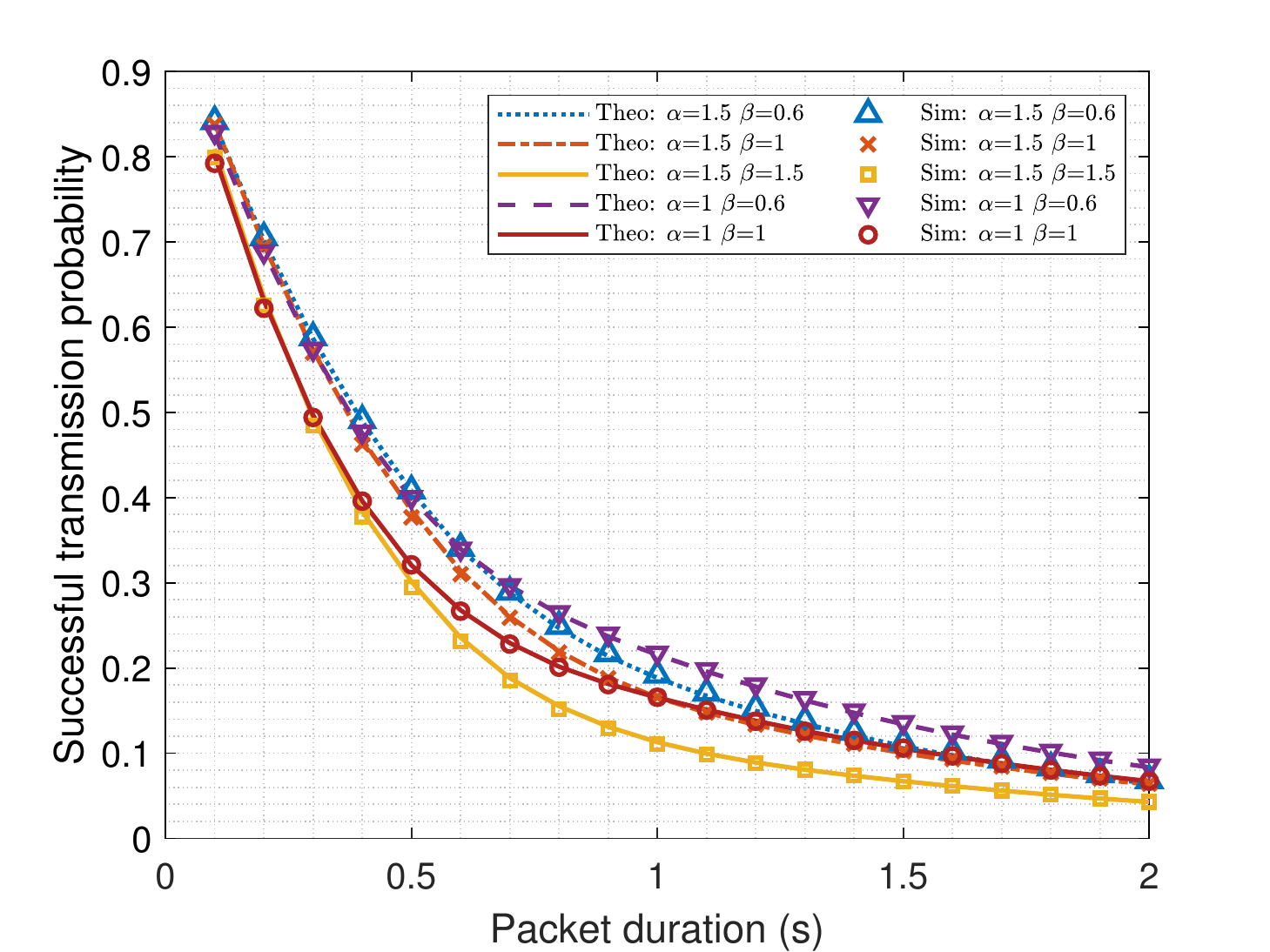}
\label{fig_Validation_VaryTf_ps}}
\hfil
\centering
\subfloat[Network throughput versus packet duration.]{\includegraphics[width=3.3in]{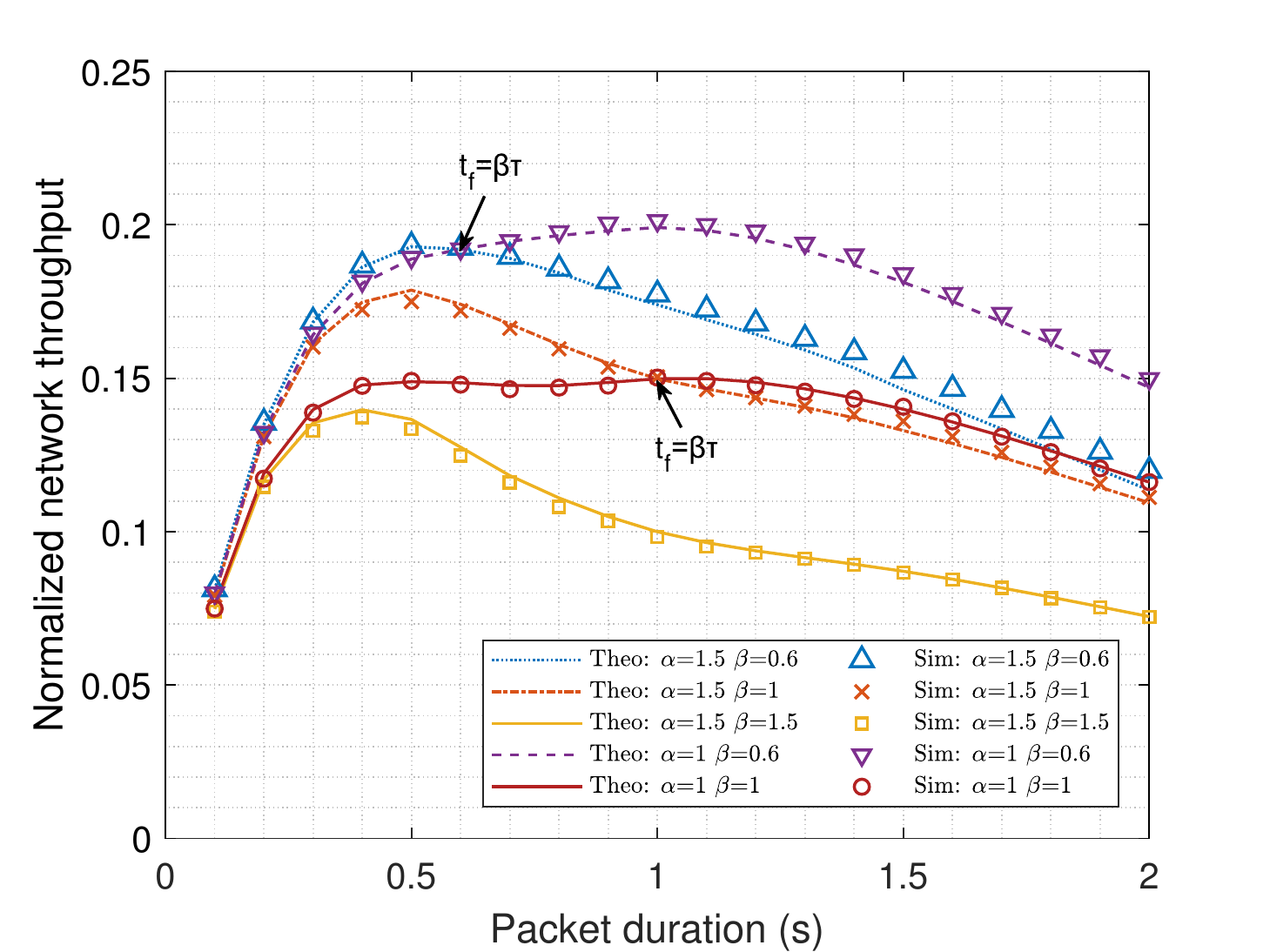}
\label{fig_validation_VaryTf_T}}
\caption{Impact of packet duration on the performance of Slotted-ALOHA.
}
\label{fig_Validation_VaryTf}
\end{figure*}

\begin{figure*}[t]
\centering
\subfloat[Successful transmission probability versus the value of $\beta$.]{\includegraphics[width=3.3in]{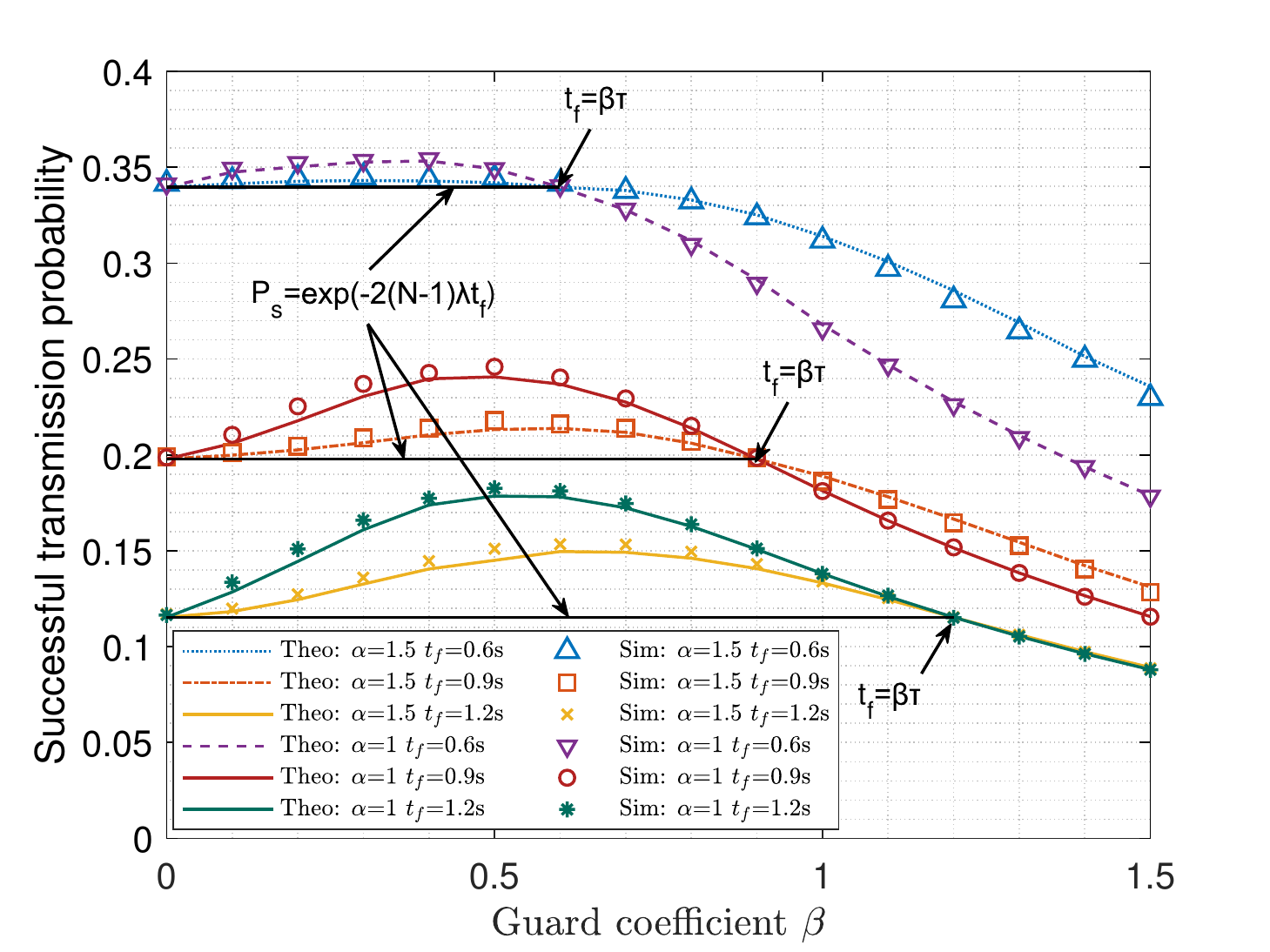}
\label{fig_Validation_VaryBeta_ps}}
\hfill
\subfloat[Network throughput versus the value of $\beta$.]{\includegraphics[width=3.3in]{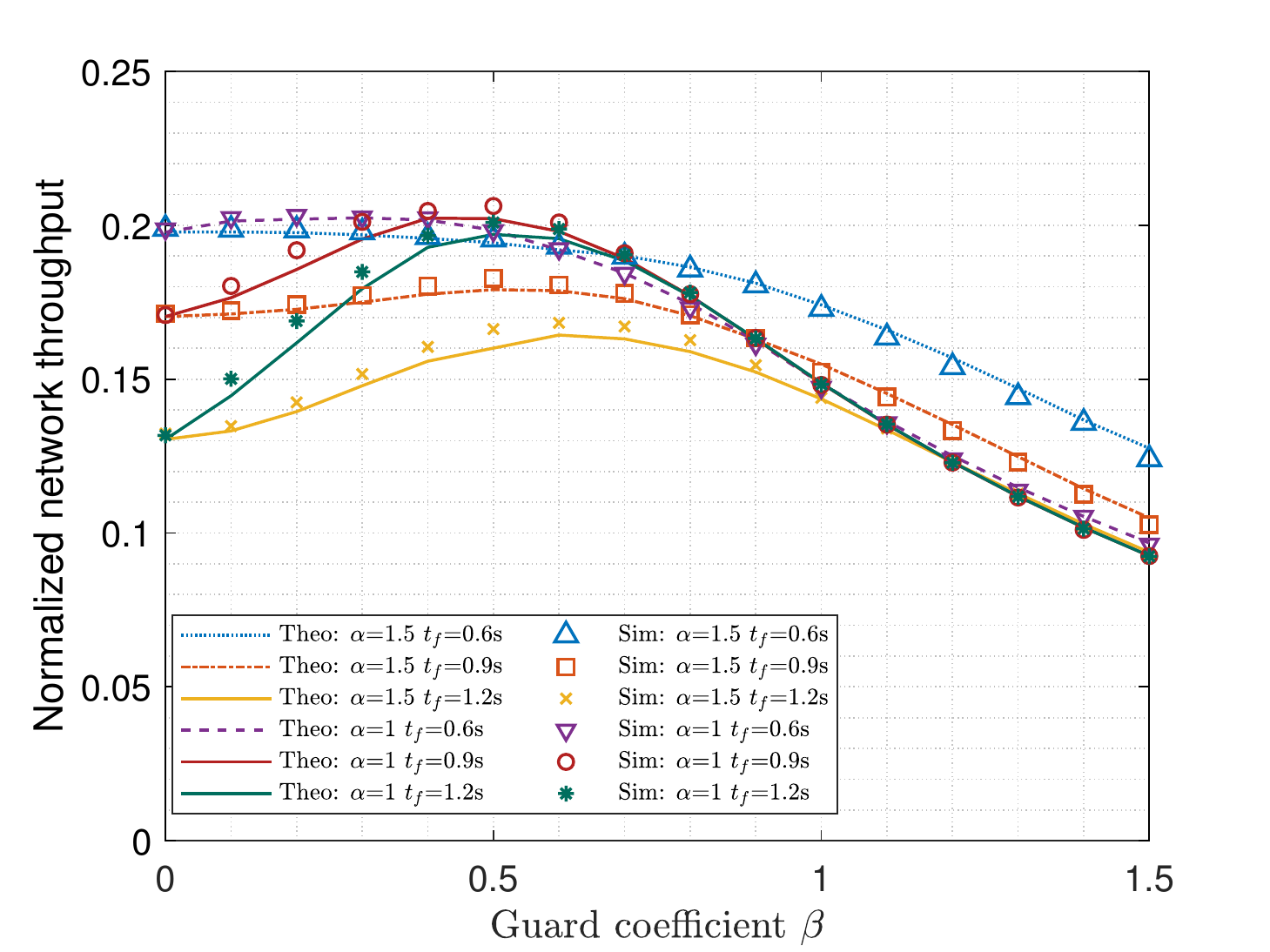}
\label{fig_Validation_VaryBeta_T}}
\caption{Impact of $\beta$ on the performance of Slotted-ALOHA.
}
\label{fig_Validation_VaryBeta}
\end{figure*}

\subsubsection{Impact of packet duration}
In the simulations, the number of underwater sensor nodes deployed in the network is 10, and the packet arrival rate for each node is $0.1$ packet/second.   Note that $\beta$ has to be not less than 1.5 for the vertical transmissions to eliminate inter-slot collisions. The coefficient $\beta$ is set to 0.6, 1, and 1.5 to show the cases with and without inter-slot collisions. 

The simulation results  in Fig. \ref{fig_Validation_VaryTf} fit well with the analytical results in \eqref{equa-total-successful-rate} and \eqref{equa_throughput}. Basically, in Fig. \ref{fig_Validation_VaryTf_ps}, the successful transmission probability monotonously decreases as the packet duration increases. On the one hand, a long packet duration has a long collision period. On the other hand, according to Fig. \ref{fig-interfering-segement}, a long packet duration will narrow CFR and move towards DIR. By comparing the curves corresponding to $\beta=0.6,1,1.5$ in Fig. \ref{fig_Validation_VaryTf}, it is observed that the slot length setting in the traditional Slotted-ALOHA (i.e., $t_{slot}=t_f+\tau$ for the horizontal channel or $t_{slot}=t_f+1.5\tau$ for the vertical channel) is not necessary to achieve a high successful transmission probability. This is because the long slot length postpones the transmission opportunities of nodes, and aggravates the channel competition in each slot.

We observe from Fig. \ref{fig_validation_VaryTf_T} that there exists an optimal packet duration to reach the peak throughput, which also verifies the motivation in \cite{Mandar-Variable-Duration} to adjust the packet duration. For the same $\alpha$, the throughput is large when $\beta$ is small. For example, the peak throughput in the case of $\alpha=1.5, \beta=0.6$ is achieved 35\% larger than that in the case of $\alpha=1.5, \beta=1.5$, which also shows the unnecessary of avoiding inter-slot collisions by using a long slot length in the space-time coupling UAN. 

It is worth noting in Fig. \ref{fig_Validation_VaryTf} that, when $t_{f}=\beta\tau$, the performance of vertical transmissions are same as horizontal transmissions (c.f. Proposition \ref{proposition6}). In this case, all the IRs just cover the whole communication range, and neither CFR nor DIR exists. 
We further observe that the vertical transmissions have better performances when the $t_{f}<\beta\tau$ and worse performances when $t_{f}>\beta\tau$ than the horizontal transmissions.
 The CFR and DIR have greater positive reward and negative penalty on the vertical transmission than on the horizontal transmission. This indicates that the vertical transmission is more sensitive to the spatial impact than the horizontal transmission.

\begin{figure*}
\centering
\subfloat[Successful transmission probability versus the amount of sender $N$.]{\includegraphics[width=3.3in]{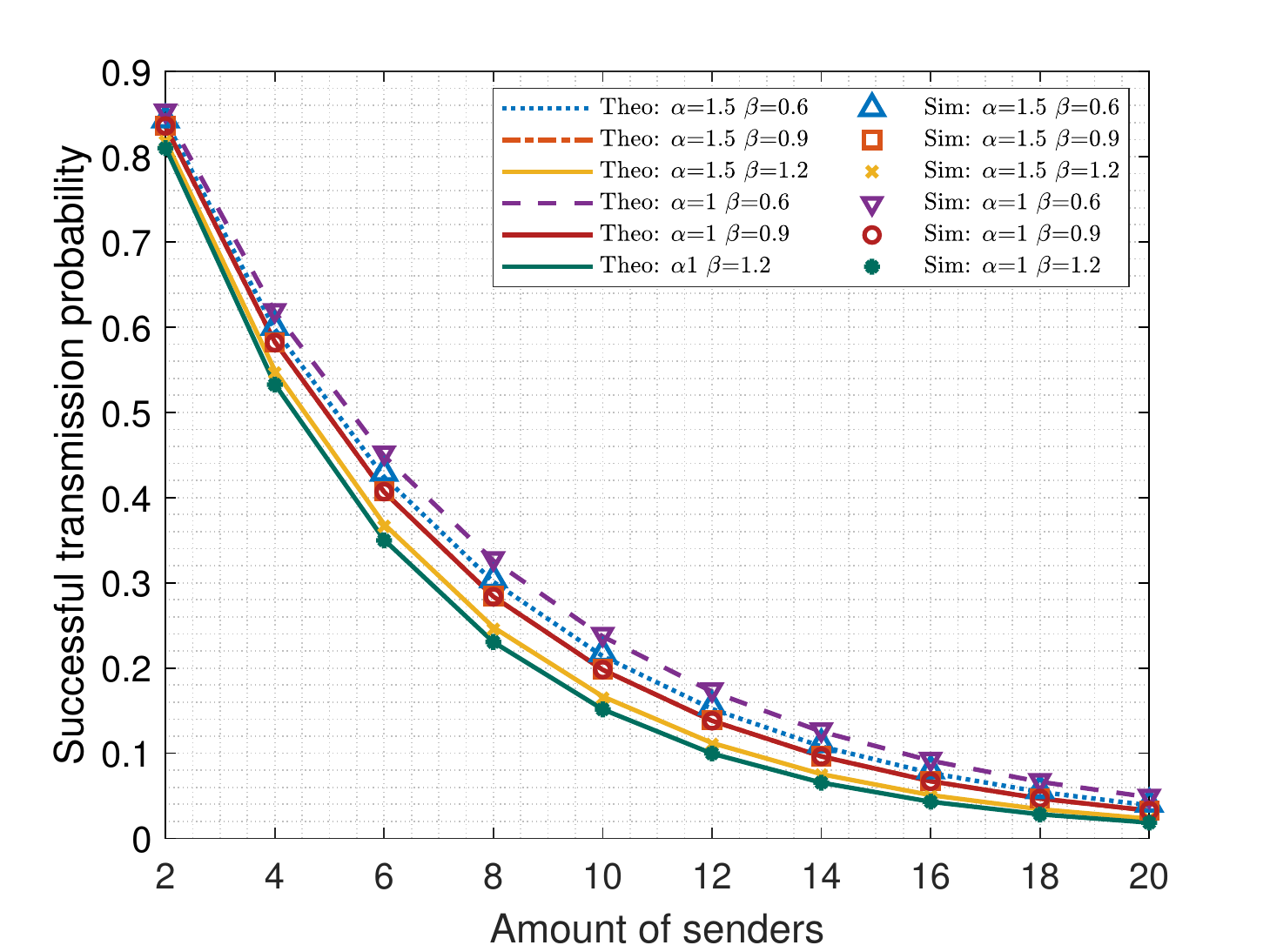}
\label{fig_Validation_VaryNt_ps}}
\hfill
\subfloat[Network throughput versus the amount of sender $N$.]{\includegraphics[width=3.3in]{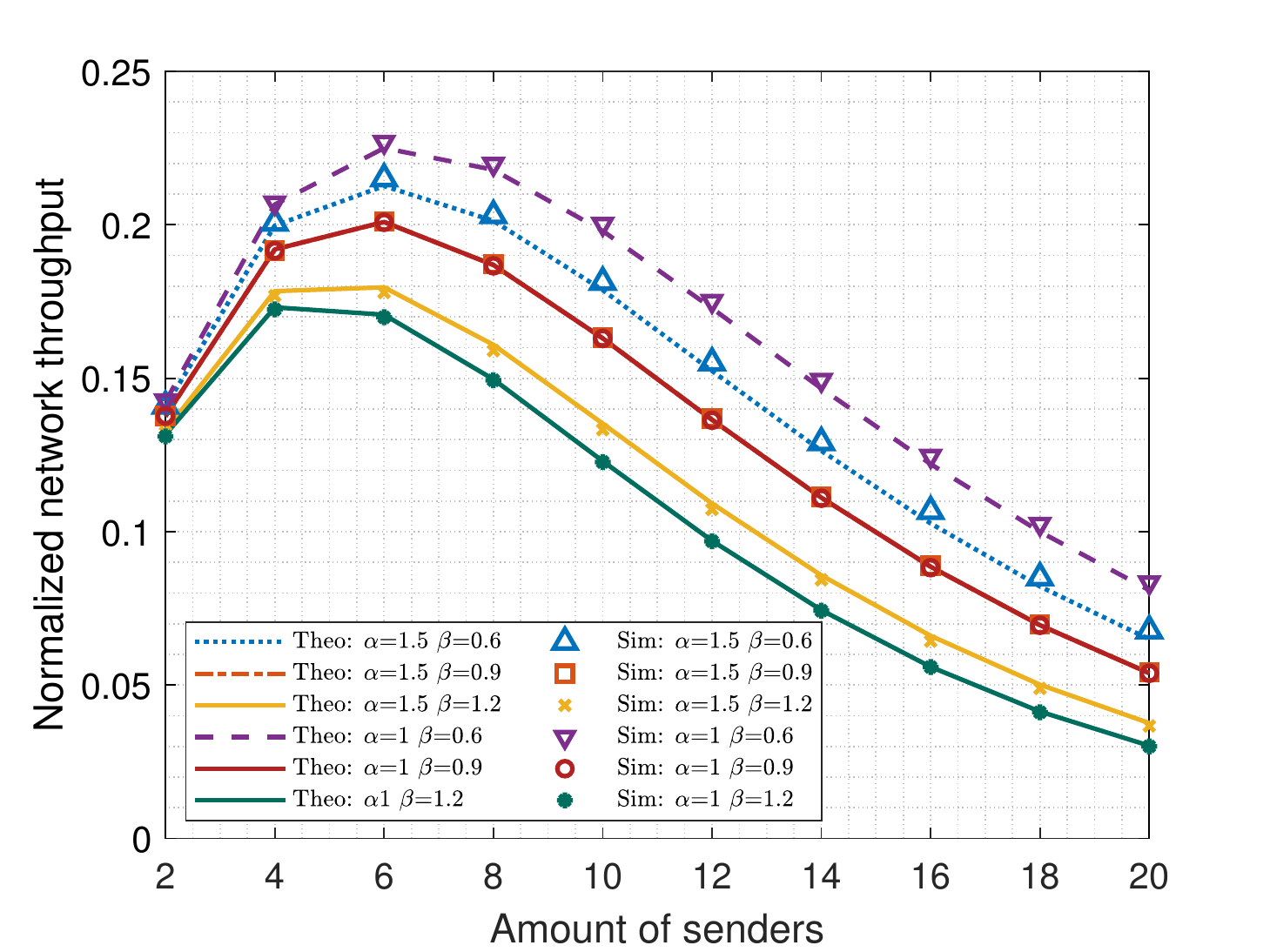}
\label{fig_Validation_VaryNt_T}}
\caption{Impact of network density on the performance of Slotted-ALOHA.
}
\label{fig_Validation_VaryNt}
\end{figure*}

\begin{figure*}
\centering
\subfloat[Successful transmission probability versus traffic load $\lambda$.]{\includegraphics[width=3.3in]{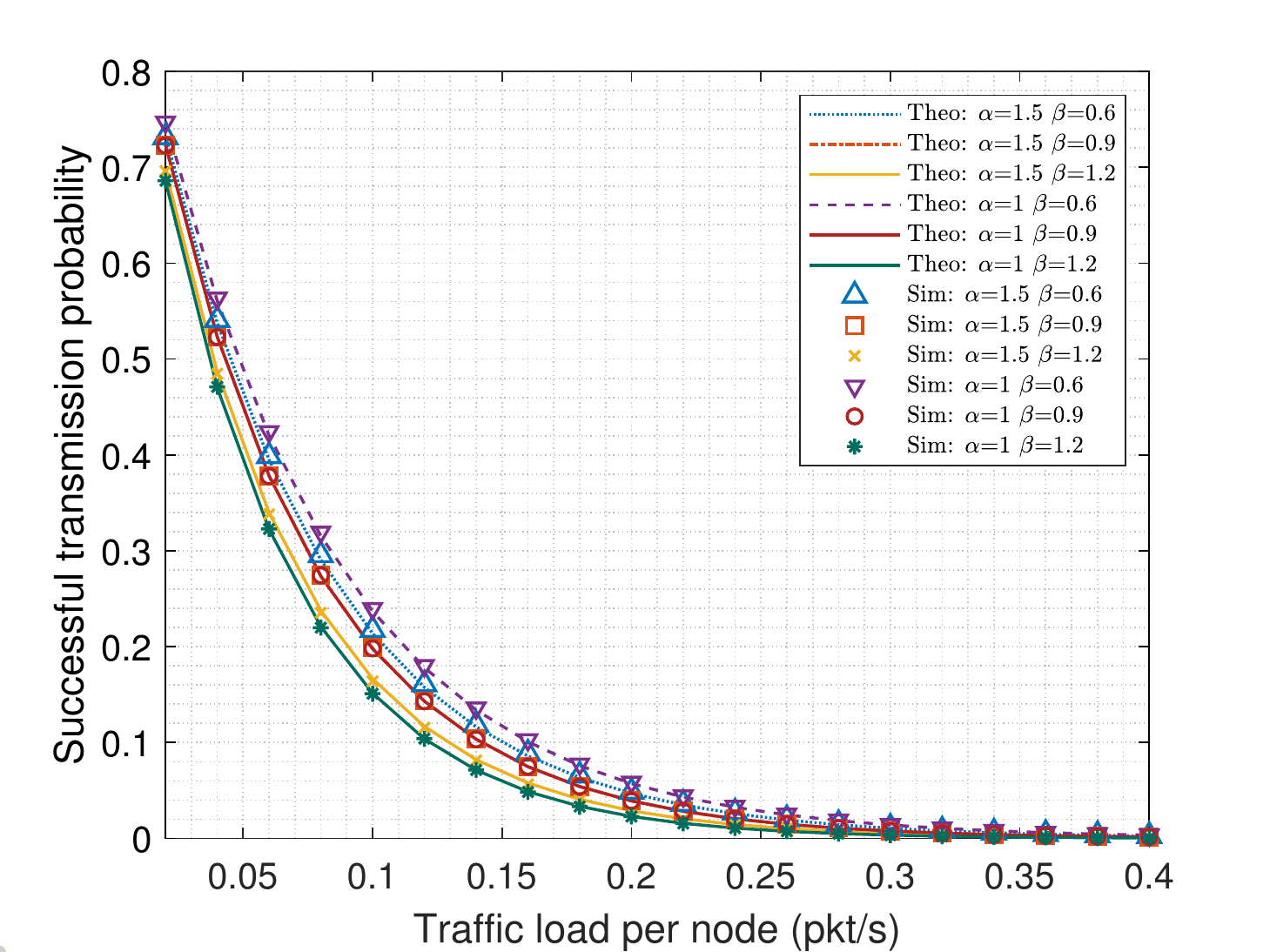}
\label{fig_Validation_VaryNload_ps}}
\hfil
\subfloat[Network throughput versus traffic load $\lambda$.]{\includegraphics[width=3.3in]{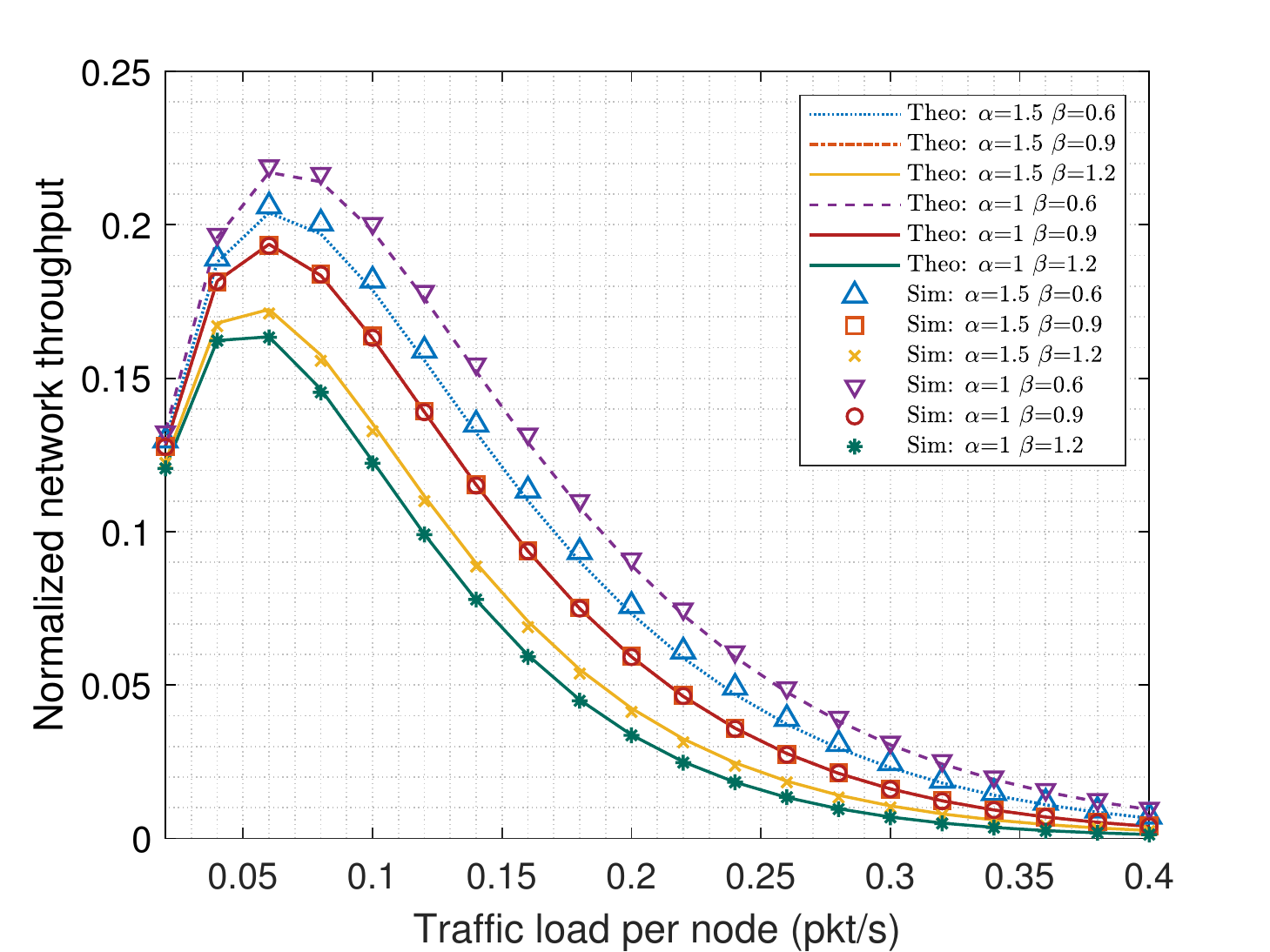}
\label{fig_Validation_VaryNload_T}}
\caption{Impact of network traffic loads on the performance of Slotted-ALOHA. 
}
\label{fig_Validation_VaryNload}
\end{figure*}

\subsubsection{Impact of $\beta$}
The coefficient $\beta$ determines the guard length and the slot length. It is observed from Fig. \ref{fig_Validation_VaryBeta} that the successful transmission probability and throughput are the same both when $\beta=0$ and $\beta\tau=t_f$, which verifies our Proposition \ref{proposition5}. Furthermore, there exists a maximizer for the guard interval in between ($0, t_f$). The successful transmission probability and throughput decrease when $\beta\tau>t_f$, which verifies our motivation that a long guard interval degrades the MAC performance.

The results in Fig. \ref{fig_Validation_VaryBeta} also show that the guard interval is unnecessary for a small packet, e.g., $t_f=0.6$ second in the figure. This is because the transmission duration mainly determines the collision period and the size of IRs as shown in Fig. \ref{fig-interfering-segement}.  A small transmission duration has already achieved a high successful transmission probability without guard interval (i.e., $\beta=0$). However, it necessitates the choice of the optimal $\beta$ to achieve the largest successful transmission probability and the highest network throughput as $t_f$ increases.  

Similar to Fig. \ref{fig_Validation_VaryTf}, it can be also observed that the vertical transmission is more sensitive to the guard interval $\beta$ than the horizontal transmission. 

\subsubsection{Impact of node density}
The increasing number of nodes in the network will aggravate the competition in channel sharing. In the next simulation, we set $\lambda=0.1$ packet/second, $t_{f}=\SI{0.9}{second}$, and $\beta\in \{0.6,0.9,1.2\}$ to study the impact of node density on network throughput. 

Fig. \ref{fig_Validation_VaryNt} shows that the successful transmission probability decreases with the increasing network density, due to the increasing channel competition. The throughput reaches its peak when the number of nodes is around 6. The reason is that the traffic load in the network is light when the number of nodes is small, while the channel gets congested as the number of nodes increases.



\subsubsection{Impact of traffic load}
We next study the impact of traffic loads. We fix the number of nodes to 10, and we set $t_{f}=\SI{0.9}{second}$ and $\beta\in \{0.6,0.9,1.2\}$. 

Similar results are observed in Fig. \ref{fig_Validation_VaryNload} as in Fig. \ref{fig_Validation_VaryNt}. That is, the increase of traffic loads will decline the successful transmission probability and there exists a peak throughput at a saturated network load of 0.07 packet/second per node. As the traffic load continues to increase, the network becomes congested. The transmissions then fail due to collisions, and the network throughput decreases to almost zero. 



\section{Conclusion}\label{sect:con}
We analyzed the impact of a unique feature of space-time coupling in UANs on slotted MAC in this paper. 
In contrast to TRNs, we found that the long propagation delay may result in transmission collisions over multiple slots, depending on the coupling factors in space and time dimensions, including the slot length, the packet duration, the sending slot, and the spatial locations of nodes. There exist not only intra-slot collisions but also inter-slot collisions in UANs, which degrades slotted MAC into pure MAC. This paper also found that the slot-dependent interference region could be annulus around the receiver. Collision-free regions (CFRs) exist when the slot length is larger than two times of the packet duration. 
This paper found that the successful transmission probabilities are the same when the slot length is set to one packet duration and is set to two packet durations. The peak successful transmission probabilities could be achieved by a slot less than two packet durations, which is much smaller than the existing slot setting. Although there exists inconsistency in vertical and horizontal transmissions, their performances are exactly the same when neither CRF and nor DIR exists by a slot length of two packet durations. Simulations have been carried out to verify the findings in this paper. 
Simulations also found that the long and non-isotropic coverage have greater positive reward and negative penalty on the vertical transmission than on the horizontal transmission. The findings in this paper will provide elaborative guidelines for the MAC design in UANs.

%

\bibliographystyle{IEEEtran} 
\bibliography{journal_ref}
%

\appendices

\section{Calculations for $S^{\Delta m}_{1,k}$ and $S^{\Delta m}_{2,k}$} \label{appendix:S}
Due to the anisotropic feature in the horizontal and vertical transmissions, we define 
\begin{equation}
\widetilde{S}(r)=\left\{\begin{matrix}
  \pi r^2,&0\leq r\leq R, \\
  2\alpha R^{2}\arcsin{\frac{1}{R}\sqrt{\frac{r^{2}-R^{2}}{\alpha^{2}-1}}}+ & \\
  2r^{2}\arccos{\frac{\alpha}{r}\sqrt{\frac{r^{2}-R^{2}}{\alpha^{2}-1}}},&R< r\leq \alpha R.
\end{matrix}\right.
\end{equation}
The following function is then introduced, 
\begin{equation}
S_{o}(a,b)=\widetilde{S}(b)-\widetilde{S}(a),
\end{equation}
where $0\leq a\leq b\leq \Lambda$, and $\Lambda$ equals to $R$ in horizontal transmission and equals to $\alpha R$ in vertical transmission.

The area of an IR or a DIR can be calculated using $S_{o}(a,b)$:
\begin{itemize}
\item For the area of an IR, we have $S^{\Delta m}_{1,k}=S_{o}(a_{1},b_{1})$, where $a_{1}=\max(d_{i}+v(\Delta m\cdot t_{slot}-t_{f}),0)$ and $b_{1}=\min(d_{i}+v(\Delta m\cdot t_{slot}+t_{f}),\Lambda)$.
\item For the area of a DIR, we have $S^{\Delta m}_{2,k}=S_{o}(a_{2},b_{2})$, where $a_{2}=\max(d_{i}+v( (m+1)\cdot t_{slot}-t_{f}),0)$ and  $b_{2}=\min(d_{i}+v(\Delta m\cdot t_{slot}+t_{f}),\Lambda)$.
\end{itemize}

%

\end{document}